\documentclass[10pt]{article}

\usepackage{booktabs}  
\usepackage{subcaption}

\usepackage{amsfonts}
\usepackage{amssymb}
\usepackage{amsmath}
\usepackage{amsthm}
\usepackage{verbatim}
\usepackage{algorithmic}
\usepackage{xcolor}
\usepackage{graphicx}
\usepackage{relsize}
\usepackage{scalerel}

\theoremstyle{plain}
\newtheorem{theorem}{Theorem}[section]
\newtheorem{lemma}[theorem]{Lemma}

\newtheorem{proposition}[theorem]{Proposition}
\theoremstyle{definition}
\newtheorem{definition}[theorem]{Definition}

\newtheorem{example}[theorem]{Example}

\begin{document}

%% Title information
\title{Weighted model counting beyond two-variable logic}

\author{Antti Kuusisto, Carsten Lutz\\
University of Bremen}

\date{}

\maketitle

\begin{abstract}
  It was recently shown by van den Broeck at al.\ that the symmetric
  weighted first-order model counting problem (WFOMC) for sentences of
  two-variable logic $\mathrm{FO}^2$ is in polynomial time, while it is
  $\#{\mathrm P}_1$-complete for some $\mathrm{FO}^3$-sentences.  We extend the result for
  $\mathrm{FO}^2$ in two independent directions: to sentences of the
  form $\varphi \wedge \forall x \exists^{{=}1} y \, \psi(x,y)$ with $\varphi$
  and $\psi$ formulated in $\mathrm{FO}^2$ and to sentences of 
  the uniform one-dimensional fragment $\mathrm{U_1}$ of FO, a
  recently introduced extension of two-variable logic with the capacity to deal with relation
  symbols of all arities. We note that the former generalizes the
  extension of $\mathrm{FO}^2$ with a functional relation symbol.  We
  also identify a complete classification of first-order prefix classes
  according to whether WFOMC is in polynomial time or $\#\mathrm{P}_1$-complete.
\end{abstract}

\section{Introduction}

The first-order model counting problem asks, given a sentence
$\varphi$ and a number $n$, how many models of $\varphi$ of size $n$ 
exist. (The domain of the models is taken to be $\{0,\dots , n-1\}$.) 
The weighted variant of this problem adds weights to atomic
facts $R^{\mathfrak{M}}(u_1,\dots, u_k)$ of models $\mathfrak{M}$, the
total weight of $\mathfrak{M}$ being the product of the atomic
weights.  The question is then what the sum of the weights of all models of
$\varphi$ of size~$n$ is. Following \cite{kimmig}, we also admit
weights of negative facts `\emph{not} $R^{\mathfrak{M}}(u_1,\dots,
u_k)$'.
% as this has, in addition to practical advantages, some striking technical benefits.

%
We investigate the symmetric weighted model counting problem of
systems extending the two-variable fragment $\mathrm{FO}^2$ of
first-order logic $\mathrm{FO}$. The word `symmetric' indicates that
each weight is determined by the relation symbol of the (positive or
negative) fact and thus the weights can be specified by weight
functions $w$ and $\bar{w}$ that assign
weights to each relation symbol occurring positively ($w$) or 
negatively ($\bar{w}$). We let $\mathrm{WFOMC}$ refer to the symmetric
weighted first-order model counting problem, with $\mathrm{WFOMC}(\varphi,n,w,\bar{w})$
denoting the sum of the weights of models $\mathfrak{M}\models\varphi$
of size $n$ according to the weight functions $w$ and $\bar{w}$.
We focus on studying the \emph{data complexity} of $\mathrm{WFOMC}$,
that is, the complexity of determining
$\mathrm{WFOMC}(\varphi,n,w,\bar{w})$ where $n$ is the only input,
given in unary, and with $\varphi,w,\bar{w}$ fixed. 
% While our approach
% can also be used to obtain results for combined complexity, where all
% of $\varphi,n,w,\bar{w}$ are inputs, we only discuss data complexity
% explicitly to keep the exposition of results brief.
%

%
The recent article \cite{DBLP:conf/kr/BroeckMD14} established the by
now well-known result that the data complexity of $\mathrm{WFOMC}$ is
in polynomial time for formulae of $\mathrm{FO}^2$, while \cite{suc15}
demonstrated that the three-variable fragment $\mathrm{FO}^3$ contains
formulae for which the problem is $\#\mathrm{P}_1$-complete. We note
that the non-symmetric variant of the problem is
known to be $\#\mathrm{P}$-complete for some
$\mathrm{FO}^2$-sentences \cite{suc15}.
Weighted model counting problems have a range of well-known applications. For
example, as pointed out in \cite{suc15}, $\mathrm{WFOMC}$ 
problems occur in a natural way in \emph{knowledge bases with soft
  constraints} and are especially prominent in the area of Markov
logic \cite{DBLP:series/synthesis/2009Domingos}.\
For a recent comprehensive survey on these matters, see
\cite{suciusurvey}.
From a mathematical perspective,
%On the other hand,  a purely theoretical justification for 
$\mathrm{WFOMC}$ %is that it
offers a neat and general
approach to \emph{elementary enumerative combinatorics}. To give a 
simple illustration of this, consider $\mathrm{WFOMC}(\varphi,n,w,\bar{w})$
for the two-variable logic sentence $\varphi =  \forall x \forall y( Rxy\rightarrow (Ryx\wedge x\not= y))$
with $w(R) = \bar{w}(R) = 1$. The sentence states that $R$ encodes a simple
undirected graph and thus $\mathrm{WFOMC}(\varphi,n,w,\bar{w}) = 2^{\binom{n}{2}}$,
the number of graphs of order $n$ (with the set $n$ of vertices).
Thus $\mathrm{WFOMC}$ 
provides a \emph{logic-based way of classifying combinatorial problems}. For instance,
the result for $\mathrm{FO}^2$-properties
from \cite{DBLP:conf/kr/BroeckMD14} shows that all these
properties can be associated with tractable enumeration functions.
For discussions of the links between weighted model counting, the
spectrum problem and 0-1 laws, see \cite{suc15}.

In the current paper, we extend the result of
\cite{DBLP:conf/kr/BroeckMD14} for $\mathrm{FO}^2$ in two independent
directions. We first consider $\mathrm{FO}^2$ with a
\emph{functionality axiom}, that is, sentences of type $\varphi\,
\wedge \forall x\exists^{=1}y\, \psi(x,y)$ with $\varphi$ and $\psi$ in
$\mathrm{FO}^2$. This extension is motivated, inter alia, by certain
description logics with \emph{functional roles}
\cite{DBLP:books/daglib/0041477}. %However,
The connection of $\mathrm{WFOMC}$ to enumerative
combinatorics %can be seen as the main
also provides an important part of the motivation. Indeed, while $\mathrm{FO}^2$ is a
reasonable formalism for specifying properties of relations, adding
functionality axioms allows us to also express properties of
functions, possibly combined with relations.
For example, applying WFOMC to the sentence $\forall x \neg Rxx
\wedge \forall x\exists^{=1}yRxy$ gives the number of
functions that do not have a fixed point.  While the extension of
$\mathrm{FO}^2$ with a functionality axiom might appear simple at
first sight, showing that the data complexity of $\mathrm{WFOMC}$
remains in $\mathrm{PTIME}$ requires a rather different and much more
involved approach than that for $\mathrm{FO}^2$.
Our proofs provide concrete and insightful aritmetic expressions for
analysing the related weighted model counts. The article \cite{kazemi} considers
weighted model counting of an orthogonal extension of $\mathrm{FO}^2$
which can express that some relations are functions.
%at least the below arguments showing the
%tractability of its $\mathrm{WFOMC}$ involve reasonably complex constructions. 

%the below arguments  
%reasonably involved arguments to establish tractability for it.

%
We also show that the data complexity of $\mathrm{WFOMC}$ remains in
$\mathrm{PTIME}$ for sentences of the \emph{uniform one-dimensional
  fragment} $\mathrm{U}_1$.  This is a recently introduced
\cite{kuusihella, kuuki} extension of $\mathrm{FO}^2$ that preserves
$\mathrm{NEXPTIME}$-completeness of the satisfiability problem while
admitting more than two variables and thus being able to speak about
relations of all arities in a meaningful way.
%$\mathrm{U}_1$ is based on
%
%
%
The fragment $\mathrm{U}_1$ is obtained from FO
by restricting quantification to blocks of existential (universal) quantifiers
that leave at most one variable free, a restriction 
referred to as the \emph{one-dimensionality} condition.
Additionally, a \emph{uniformity condition} is imposed:
if $k,n\geq 2$, then a Boolean combination of atoms
$Rx_1\dots x_k$ and $Sy_1\dots y_n$ is allowed only if 
the sets $\{x_1,\dots , x_k\}$ and $\{y_1,\dots , y_n\}$ of variables are equal.
Boolean combinations of formulae with at
most one free variable can be formed freely,
and the use of equality is unrestricted. 
%
% clu: holds always and for everything, right?
%
%See Section \ref{preliminaries} for formal details.
%
It is shown in \cite{kuusihella} that lifting either of these
conditions---in a minimal way---leads to undecidability.
For a survey of the basic properties of $\mathrm{U}_1$ and its
relation to modal and description logics, see~\cite{kuusisurvey}.

What makes weighted model counting for $\mathrm{U}_1$ attractive in
relation to applications is the ability of $\mathrm{U}_1$ to express interesting
properties of relations of all arities, thereby banishing one of the main
weaknesses of $\mathrm{FO}^2$. This is especially well justified
from the points of view of database theory and of knowledge
representation with formalisms such as Markov logic, which are among
the main application areas of $\mathrm{WFOMC}$. We 
note that $\mathrm{U}_1$ is significantly more expressive than
$\mathrm{FO}^2$ already in restriction to models with at most binary
relations~\cite{kuusisurvey}.

We also identify a complete classification of first-order prefix
classes according to whether the sentences of the particular class 
have polynomial time $\mathrm{WFOMC}$ or whether some sentence of the
class has a $\#\mathrm{P}_1$-complete $\mathrm{WFOMC}$. This
classification, whose proof makes significant use of the
results and techniques from \cite{suc15,DBLP:conf/kr/BroeckMD14}, is
remarkably simple: $\#\mathrm{P}_1$-hardness arises precisely for the
classes with more than two quantifiers, independently of the quantifier
\nolinebreak 
pattern.

\begin{comment}
The expressivity of $\mathrm{FU}_1$ is \emph{precisely} the
same as that of $\mathrm{FO}^2$ in
restriction to vocabularies with arities at most two. $\mathrm{U}_1$ is 
the extension of $\mathrm{FU}_1$ with free use of equality.
\end{comment}

\begin{comment}
\cite{suc15} \cite{DBLP:conf/kr/BroeckMD14}

\cite{kieku15} \cite{kuuki} \cite{kuusisurvey}

\cite{ebbinghausflum}

\cite{kuusihella}

\cite{jonne}

\cite{guardednegation}
\end{comment}

%
\section{Preliminaries}\label{preliminaries}
The natural numbers are denoted by $\mathbb{N}$
and positive integers by $\mathbb{Z}_+$.
As usual, we often identify $n\in\mathbb{N}$ with the
set $\{\, k\in\mathbb{N}\ |\ k < n\}$.
We define $[n] := \{1,\dots , n\}$ for each $n\in\mathbb{Z}_+$ 
\textcolor{black}{and $[0] = \emptyset$}.
The domain of a function $f$ is denoted by $\mathit{dom}(f)$.
The function $f$ is \emph{involutive} if $f(f(x)) = x$ 
for all $x \in \mathit{dom}(f)$ and 
\emph{anti-involutive} if $f(f(x))\not= x$ for all $x \in \mathit{dom}(f)$.
Two functions $f$ and $g$ are \emph{nowhere inverses} if $f(g(x))\not=x$ 
and $g(f(y))\not = y$ for all $x\in \mathit{dom}(g)$, $y\in\mathit{dom}(f)$.
We use the standard notation $\binom{n}{n_1,\dots , n_m}$ 
for multinomial coefficients. 
We study (fragments of) first-order logic $\mathrm{FO}$
over relational vocabularies; constant and function symbols are not
allowed.  The identity symbol `$=$' and the Boolean constants
$\bot,\top$ are \emph{not} considered relation symbols; they are a
logical symbols included in $\mathrm{FO}$. We allow nullary relation
symbols in $\mathrm{FO}$ with the usual syntax and semantics.
The vocabulary of a formula $\varphi$ is
denoted by $\mathit{voc}(\varphi)$.
We let $\mathrm{VAR} := \{v_0,v_1,\dots\}$ denote a fixed,
countably infinite set of
%first-order
variable symbols. We mainly use
meta-variables $x,y,z,$ etc., in order to refer to symbols in $\mathrm{VAR}$.
Note that for example $x$ and $y$ may denote the same variable, 
while $v_i$ and $v_j$ are different if $i\not=j$.
%

%
% For a positive integer $n$, we let $[n]$ denote the set $\{1,\dots,n\}$.
% The \emph{lineage} $\mathit{lng}(\varphi)$ of a first-order formula $\varphi$
% over a set $[n]$ is intuitively a translation of $\varphi$ into a quantifier-free formula
% with constants from $[n]$ such that quantifiers $\exists$ (respectively, $\forall$) are
% replaced by big disjunctions (respectively, conjunctions) in the natural way;
% see \cite{suc15} for the formal definition.
%

%
%When confusion should not arise, we do not
The domain of a model $\mathfrak{M}$ is denoted by
$\mathit{dom}(\mathfrak{M})$. In the case $A\subseteq(
\mathit{dom}(\mathfrak{M}))^k$, we let $(\mathfrak{M},A)$ denote the
expansion of $\mathfrak{M}$ obtained by adding the $k$-ary relation
$A$ to $\mathfrak{M}$. We mostly do not differentiate between
relations and relation symbols explicitly when the distinction is
clear from the context. Relational models decompose into \emph{facts}
and \emph{negative facts} in the usual way: if $R$ is a $k$-ary
relation symbol of a model $\mathfrak{M}$ and $Ru_1\dots u_k$ holds
for some elements $u_1,\dots,u_k$ of $\mathfrak{M}$, then $Ru_1\dots
u_k$ is a positive fact of $\mathfrak{M}$, and if $Ru_1\dots u_k$ does
not hold in $\mathfrak{M}$, then $Ru_1\dots u_k$ is a negative fact of
$\mathfrak{M}$. We denote the positive (respectively, negative) facts
of $\mathfrak{M}$ by $F^+(\mathfrak{M})$ (respectively,
$F^-(\mathfrak{M})$).  The \emph{span} of a fact $Ru_1,\dots ,u_k$,
whether positive or negative, is $\{u_1,\dots , u_k\}$ and its
\emph{size} is $|\{u_1,\dots , u_k\}|$.
The \emph{first-order model counting problem} asks, when 
given a positive integer $n$ in \emph{unary} and an $\mathrm{FO}$-sentence $\varphi$,
how many models $\varphi$ has over the domain $n = \{0,\dots , n-1\}$;
the vocabulary of the models is taken to be $\mathit{voc}(\varphi)$, and
different but isomorphic models contribute separately to the output.
The \emph{weighted first-order model counting problem}
adds two functions to the input, $w$ and $\bar{w}$, that both map 
the set of all possible facts over \textcolor{black}{$n$}
and $\mathit{voc}(\varphi)$ into a 
set of weights. In the \emph{symmetric} weighted model counting problem
studied in this paper, $w$ and $\bar{w}$ are functions $w:\mathit{voc}(\varphi)
\rightarrow \mathbb{Q}$
and $\bar{w}:\mathit{voc}(\varphi)\rightarrow \mathbb{Q}$.
The output $\mathrm{WFOMC}(\varphi,n,w,\bar{w})$ is then
the sum of the \emph{weights} $W(\mathfrak{M},w,\bar{w})$ of
all models $\mathfrak{M}\models\varphi$ with
domain $n$ and vocabulary $\mathit{voc}(\varphi)$,
\begin{equation}
W(\mathfrak{M},w,\bar{w}) := \prod\limits_{Ru_1\dots u_k\, \in\, F^+(\mathfrak{M})}w(R)\ \mathbf{\boldsymbol{\cdot}}\ \prod
\limits_{Ru_1\dots u_k\, \in\, F^-(\mathfrak{M})}\bar{w}(R).
\end{equation}
This setting gives rise to several computational problems, depending on
which inputs are fixed. In this article, we
exclusively study \emph{data complexity}, i.e., the 
problem of computing $\mathrm{WFOMC}(\varphi, n, w,\bar{w})$ 
with the sole input $n\in\mathbb{Z}_+$ given in unary; $\varphi$, $w$ and $\bar{w}$ 
are fixed and thus not part of the input. Algorithms for more general inputs can easily be
extracted from our proofs, but we only study data complexity explicitly
for the lack of space.
While weights are rational numbers,
it will be easy to see that reals with a tame enough representation
could also be included without sacrificing our results. We ignore this 
for the sake of simplicity and stick to rational weights. (See also \cite{kimmig}.)
We now define, for technical purposes, some restricted
versions of $\mathrm{WFOMC}$ and the operator $\mathrm{W}$. First, if $\mathcal{M}$ is a
class of models, we
define $$\mathrm{WFOMC}(\varphi,n,w,\bar{w})\upharpoonright\mathcal{M}$$ to
be the sum of the weights $\mathrm{W}(\mathfrak{M},w,\bar{w})$ of 
models $\mathfrak{M}\in\mathcal{M}$ with domain $n$ and
vocabulary $\mathit{voc}(\varphi)$ such that $\mathfrak{M}\models\varphi$.
For $k\in\mathbb{Z}_+$,
we let $F^+_k(\mathfrak{M})$ and $F^-_k(\mathfrak{M})$ denote the
restrictions of $F^+(\mathfrak{M})$ and $F^-(\mathfrak{M})$ to
facts with span of size $k$. We define
$W_k(\mathfrak{M},w,\bar{w})$ exactly as $W(\mathfrak{M},w,\bar{w})$ but
with $F^+(\mathfrak{M})$ and $F^-(\mathfrak{M})$
replaced by $F_k^+(\mathfrak{M})$ and $F_k^-(\mathfrak{M})$.
When $\varphi, n, w$ and $\bar{w}$ are clear from the context,
we use the \emph{weight of a
class $\mathcal{M}$ of models}
to refer to $\mathrm{WFOMC}(\varphi,n,w,\bar{w})\upharpoonright\mathcal{M}$.
%
%
%
\begin{comment}
\begin{equation}\label{relativisedweight}
W_m(\mathfrak{M},w,\bar{w}) := \prod\limits_{Ru_1\dots u_k\, \in\, F_m^+(\mathfrak{M})}w(R)\ \mathbf{\cdot}\ \prod
\limits_{Ru_1\dots u_k\, \in\, F_m^-(\mathfrak{M})}\bar{w}(R).
\end{equation}
\end{comment}
%
%
%

The quantifier-free part of a prenex normal form formula of $\mathrm{FO}$ is called a
\emph{matrix}. A prenex normal form sentence of type $\chi := \forall
x_1\dots \forall x_k\psi$, where $\psi$ is the matrix, is a
\emph{$\forall^*$-sentence}, and the number $k$ of quantifiers in
$\chi$ is the \emph{width} of $\chi$.  An $\exists^*$-sentence is
defined analogously. 
%A matrix of a formula is a
%\emph{$\mathrm{U}_1$-matrix}.
%

%In addition to uniform one-dimensional fragments, we also
We will investigate standard two-variable logic $\mathrm{FO}^2$
enhanced with a \emph{functionality axiom}.  Formulae in this language
are conjunctions of the type $\varphi\wedge \forall
x\exists^{=1}y\psi(x,y)$, where $\varphi$ and $\psi$ are
$\mathrm{FO}^2$-formulae, $\psi$ with the free variables $x,y$ and
$\varphi$ a sentence.  When studying this variant of FO, we
exclusively use the variables $x,y$, with $x$ denoting $v_1$ and $y$
denoting $v_2$.

%
%Our next aim is to define the \emph{uniform one-dimensional 
%fragment} of first-order logic. We begin with some auxiliary definitions.
%

We next introduce uniform one-dimensional fragments of FO.
Let $Y= \{y_1,\dots,y_k\}$ be a
set of distinct variables, and let $R$ be an $n$-ary relation symbol for some $n\geq k$.
An atom $Ry_{i_1}\dots , y_{i_n}$ is a \emph{$Y$-atom} if $\{y_{i_1},\dots , y_{i_n}\} = Y$.
For example, if $x,y,z,v$ are distinct variable
symbols, then $Txyzx$ and $Sxzy$ are $\{x,y,z\}$-atoms,
while $Uxyzv$ and $Vxy$ are not. Furthermore, $Vxz$ is an $\{x,z\}$-atom
while $x=z$ is not as identity is not a 
relation symbol. A $Y$-literal is a $Y$-atom $Ry_{i_1}\dots , y_{i_n}$ or a
negated $Y$-atom $\neg Ry_{i_1}\dots  y_{i_n}$. A $Y$-literal is
\emph{an $m$-ary literal} if $|Y| = m$, so for example $Sxx$ and $\neg Px$
are unary literals; $Sxx$ is even a unary atom while $\neg Px$ is not.
A \emph{higher arity literal} is a literal of arity at least two.
We let $\mathit{diff}(x_1,\dots , x_k)$ denote the conjunction of 
inequalities $x_i\not= x_j$ for all distinct $i,j\in [k]$.
%We define $$\mathit{diff}(x_1,\dots , x_k) := \bigwedge\limits_{ i,j\in [k],\ i\not=j} x_i\not= x_j.$$
%We let $\mathit{diff}(x_1,\dots ,x_k)$ denote the 
%conjunction of inequalities $x_i\not= x_j$ for all distinct
%variables $x_i,x_j$ from $\{x_1,\dots , x_k\}$.

%
The set of formulae of the \emph{uniform one-dimensional
fragment} $\mathrm{U_1}$ of $\mathrm{FO}$ is the
smallest set $\mathcal{F}$ such that the following conditions hold.
\begin{enumerate}
\item
Unary and nullary atoms are in $\mathcal{F}$.
\item
All identity atoms $x=y$ are in $\mathcal{F}$.
\item
If $\varphi,\psi\in\mathcal{F}$,
then $\neg\varphi\in\mathcal{F}$ and $\varphi\wedge\psi\in\mathcal{F}$.
\item
Let $X=\{x_0,\dots,x_k\}$ and $Y\subseteq X$. Let $\varphi$ be a Boolean
combination of $Y$-atoms and formulae in $\mathcal{F}$ whose
free variables (if any) are in $X$. Then
\begin{enumerate}
\item
$\exists x_1\dots\exists x_k\, \varphi\in\mathcal{F}$,
\item
$\exists x_0\dots\exists x_k\, \varphi\in\mathcal{F}$.
\end{enumerate}
\end{enumerate}
For example $\exists y\exists z( (\neg Rxyz  \vee Tzyxx) \wedge Qy)$ is a $\mathrm{U}_1$-formula
while $\exists x \exists y(Sxy \wedge Sxz)$ is not, as $\{x,y\}\not=\{x,z\}$.
This latter formula is said to violate the uniformity condition of $\mathrm{U_1}$.
Also $\exists z\forall y\forall x(Txyz \wedge \exists u Sxu)$ is a $\mathrm{U}_1$-formula
while $\exists x \exists y \exists z(Txyz \wedge \exists u Txyu)$ is not,
as $\exists u Txyu$ leaves two variables free and thereby 
violates the one-dimensionality condition of $\mathrm{U}_1$.
The clause 4 above does not require that $Y$-atoms \emph{must} be included, so 
also $\exists x \exists y \exists z\mathit{diff}(x,y,z)$ is a $\mathrm{U}_1$-formula.
We thus see that $\mathrm{U}_1$ has some counting capacities.
A matrix of a $\mathrm{U}_1$-formula is a
called a \emph{$\mathrm{U}_1$-matrix}.
The article \cite{kuusisurvey} contains a survey of $\mathrm{U}_1$ with background
about its expressive power and connections to extended modal logics. The 
article \cite{kieku15} provides an
Ehrenfeucht-Fra\"{i}ss\'{e} game characterization of $\mathrm{U}_1$.
It is worth noting that the so-called \emph{fully uniform one-dimensional fragment} $\mathrm{FU}_1$
has \emph{exactly} the same expressive power as $\mathrm{FO}^2$ when 
restricting to vocabularies with at most binary relations \cite{kuusisurvey}.
The logic $\mathrm{FU}_1$ is obtained by dropping clause 2 from the above
definition of $\mathrm{U}_1$ and instead regarding the identity symbol as an 
ordinary binary relation in clause 4; see \cite{kuusisurvey}. 
Thus $\mathrm{U}_1$ is the extension of $\mathrm{FU}_1$ with unrestricted use of identity.
The formula $\exists x \exists y \exists z\mathit{diff}(x,y,z)$ is an obvious
example of a $\mathrm{U}_1$-formula that is not expressible in $\mathrm{FO}^2$.
Another formula worth mentioning here that separates the 
expressive powers of $\mathrm{U}_1$ and $\mathrm{FO}^2$ 
is $\exists x \forall y \forall z(Ryz \rightarrow (x=y\vee x=z))$
which states that some node is part of every edge of $R$. The separation was
shown in \cite{kuusisurvey}, and the proof is easy; simply consider the two-pebble game
(defined in, e.g., \cite{ebbinghausflum}) on
the complete graphs $K_2$ and $K_3$.
The $\mathrm{U}_1$-formula $\exists x \exists y \exists z \neg Sxyz$ is 
one of the simplest formulae separating $\mathrm{U}_1$ from \emph{both} $\mathrm{FO}^2$ 
and the guarded negation fragment
\cite{guardednegation}, as shown in \cite{kuusisurvey}.
For technical purposes, we also introduce the 
\emph{strongly restricted} fragment of $\mathrm{U}_1$,
denoted $\mathrm{SU}_1$, which was originally
introduced and studied in \cite{kieku15}. The logic $\mathrm{SU}_1$
imposes the additional condition on the above clause 4 that the set $Y$
must contain exactly all of the variables $x_0,\dots , x_k$. For
example $\exists x \exists y \exists u (Rxyu \wedge x\not= u)$ is an $\mathrm{SU}_1$-formula
while $\exists x \exists y (Sxy \wedge x\not= z)$ is not,
despite being a $\mathrm{U}_1$-formula, as $z\not\in\{x,y\}$.
Despite the syntactic restriction imposed by $\mathrm{SU}_1$ being simple, it has some
significant consequences: it is shown in \cite{kieku15} that the satisfiability
problem of $\mathrm{SU}_1$ in the presence of a \emph{single} built-in
equivalence relation is only NEXPTIME-complete,
while it is 2NEXPTIME-complete for $\mathrm{U}_1$.
We note that even the restriction $\mathrm{SU}_1$ of $\mathrm{U}_1$
contains $\mathrm{FO}^2$ as a syntactic fragment.

%
%More generally, the quantifier-free part of a 
%first-order formula $\varphi$ in prenex normal form is called the \emph{matrix} of $\varphi$
%and denoted by $\mathit{mtr}(\varphi)$.
%

%
%Let $\varphi$ be a sentence of $\mathrm{U}_1$.
A\,  $\mathrm{U}_1$-sentence $\varphi$ is in \textit{generalized Scott normal form}, if
\vspace{1mm}\\
%
%
%
%\begin{multline*}
%
%
%
$\text{ }
\ \ \ \ \ \ \ \ \ \ \ \ \ \varphi =
\bigwedge \limits_{1\leq i \leq m_{\forall}} \forall x_1 \ldots \forall x_{\ell_i}\, 
\varphi_{i}^{\forall}(x_1,\ldots,x_{\ell_i})\\
\text{ }\ \ \ \ \ \ \ \ \ \ \ \ \ \ \ \ \ \ \ \ \ \ \ \ \ 
\wedge \bigwedge \limits_{1\leq i \leq m_{\exists}} \forall x \exists y_1
\ldots \exists y_{k_i}
\varphi_{i}^{\exists}(x,y_1,\ldots,y_{k_i}),$\vspace{1mm}
%
%
%
%
%
%
%\end{multline*}
%
%
%

\noindent
where $\varphi_{i}^{\exists}$ and $\varphi_{i}^{\forall}$
are quantifier-free.
%
%
%
%Henceforth by a normal form we always mean generalized Scott normal form.
%
\begin{comment}
The formulae $\forall x \exists y_1... \exists y_{k_i}
% 
\varphi_{i}^{\exists}(x,y_1,...,y_{k_i})$ are 
%
called \emph{existential conjuncts} and the
%
formulae $\forall x_1...\forall x_{l_i}
% 
\varphi_{i}^{\forall}(x_1,...,x_{l_i})$
%
\emph{universal conjuncts} of $\varphi$.
\end{comment}
%
A sentence of\, $\mathrm{FO}^2$ is in (standard)
Scott normal form if it is of type

\smallskip

$\text{ }\ \ \ \ \ \ \ \ \ \ \ \ \ \ \ \ \ \ 
\forall x\forall y\, \varphi(x,y)
\wedge\ \ \ \ \bigwedge_{1\leq i \leq m_{\exists}}\forall x\exists y \psi_i(x,y)$

\smallskip

\noindent
with $\varphi$ and each $\psi_i$ quantifier-free.
There exists a standard procedure (see, e.g., \cite{ebbinghausflum, kuuki}) that converts
any given formula $\varphi$ of $\mathrm{FO}^2$ (respectively, $\mathrm{U}_1$) in
polynomial time into a formula $\mathit{Sc}(\varphi)$ in
standard (respectively, generalized) Scott normal form
such that $\varphi$ is equivalent to $\exists P_1\dots \exists P_n\mathit{Sc}(\varphi)$,
where $P_1,\dots , P_n$ are fresh unary and nullary predicates.
The procedure is well-known and 
used in most papers on $\mathrm{FO}^2$ and $\mathrm{U}_1$, so we here only 
describe it very briefly. See Appendix \ref{scottnormalforms} for
further details. 
The principal idea is to replace, starting from the atomic level and
working upwards from there, any
subformula $\psi(x) = Qx_1\dots Qx_k\chi$, where $Q\in\{\forall,\exists\}$ 
and $\chi$ is quantifier-free, with an atomic formula $P_{\psi}(x)$, where $P_{\psi}$
is a fresh relation symbol. This novel atom $P_{\psi}(x)$ is then 
separately axiomatized to be equivalent to $\psi(x)$.
%See Apppendix QQQ for further details.
%

%
If $\varphi$ is a sentence of $\mathrm{U}_1$ (respectively $\mathrm{SU}_1$, $\mathrm{FO}^2$),
then $\mathit{Sc}(\varphi)$ is likewise a
sentence of $\mathrm{U}_1$ (respectively $\mathrm{SU}_1$, $\mathrm{FO}^2$);
see Appendix \ref{scottnormalforms}.
%(cf. Proposition 1 of \cite{kieku15}).
%
%
%
Each novel predicate ($P_{\psi}$ in the above example) is axiomatized to be
equivalent to the subformula ($\psi(x)$ in the above example) whose quantifiers are to
be eliminated, so the interpretation of the 
predicate is fully determined by the subformula in
every model of the ultimate Scott normal form sentence.
Thus, recalling that $\varphi \equiv \exists P_1\dots
\exists P_k\mathit{Sc}(\varphi)$, where $P_1,\dots P_k$ are 
the fresh predicates, we get the 
following (see Appendix \ref{scottnormalforms}
and cf. \cite{DBLP:conf/kr/BroeckMD14}).
%

%
%\begin{lemma}\label{scottlemma}
%
%Formulae of $\mathrm{U}_1$ translate in polynomial time to $\mathrm{U}_1$-formulae in
%generalized Scott normal form.
%
%\end{lemma}
%

%
%We let $\mathit{Sc}(\varphi)$ denote the Scott-normal form variant of $\varphi$; we
%may assume, w.l.o.g., that the operator $\mathit{Sc}$ is indeed a function.
%

\hyphenation{whe-re}

\begin{lemma}\label{scottlemma2}
$\mathrm{WFOMC}(\varphi,n,w,\bar{w})
\ \ =\ \ \mathrm{WFOMC}(\mathit{Sc}(\varphi),n,w',\bar{w}')$, 
where $w'$ and $\bar{w}'$ map the fresh symbols to $1$.
\end{lemma}
%
%
%
\begin{comment}
\begin{proof}[Proof sketch]
%
The claim holds as the reduction to Scott normal form introduces 
unary predicates such that they must be uniquely interpreted.
Essentially, the uniqueness of interpretation arises from the fact that
the fresh predicates simply encode thruth sets of subformulae of the original formula (and
formulae that arise in the translation process).
%
\end{proof}
\end{comment}
%

%
\subsection{Types and tables}
Let $\eta$ be a finite relational vocabulary.
A 1-type (over $\eta$) is a maximally consistent set of $\eta$-atoms and negated
$\eta$-atoms in the single variable $v_1$. The number of $1$-types over $\eta$ is clearly finite.
%
%
% 
%We denote 1-types by $\alpha$ and the set of all 1-types by $\mathit{types}_1{\eta}$.
%
%
%
%
%
%
We often identify a $1$-type $\alpha$ with the
conjunction of its elements, whence $\alpha(v_1)$ is
simply a formula in the single variable $v_1$. While the official variable with
which $\alpha$ is defined is $v_1$, we frequently 
consider $1$-types $\alpha(x), \alpha(y)$, etc., with $v_1$
replaced by other variables. To see some examples, consider the
case where $\eta = \{R,P\}$ with $R$
binary and $P$ unary.
Then the $1$-types over $\eta$ in the variable $x$ 
are $Rxx\wedge Px$, $\neg Rxx \wedge Px$, $Rxx\wedge \neg Px$
and $\neg Rxx\wedge \neg Px$.
Let $\mathfrak{M}$ be an $\eta$-model and $\alpha$ a $1$-type over $\eta$.
%
%
%  
%The type $\alpha$ is said to be \emph{realized} in $\mathfrak{M}$
% 
%if there exists some $a \in A$ such that $\mathfrak{A} \models \alpha(a)$.
%
%
% 
An element $u\in\mathit{dom}(\mathfrak{M})$
\emph{realizes} the $1$-type $\alpha$ if\, $\mathfrak{M}\models\alpha(u)$.
Note that every element of $\mathfrak{M}$ realizes exactly one $1$-type over $\eta$.
Let $k\geq 2$ be an integer.
A \emph{$k$-table} over $\eta$
is a maximally consistent set of $\{v_1,...\, ,v_k\}$-atoms
and negated $\{v_1,...\, ,v_k\}$-atoms over $\eta$.
We define that $2$-tables do \emph{not} contain identity atoms or negated identity atoms.
For example, using $x,y$ instead of $v_1,v_2$,
the set $\{Rxxy,Rxyx,\neg Ryxx, Ryyx,\neg Ryxy, Rxyy,
Sxy, \neg Syx\}$ is a $2$-table over $\{R,S\}$, where $R$ is a ternary
and $S$ a binary symbol.
We often identify a $k$-table $\beta$  
with a conjunction of its elements.
We also often consider formulae such
as $\beta(x_1,\dots , x_k)$, thereby writing $k$-tables in
terms of variables other than $v_1,\dots , v_k$.
%
%
%
\begin{comment}
If $a_1,\dots , a_k \in A$ are \emph{distinct} elements
%
such that $\mathfrak{A} \models \beta(a_1,\ldots,a_k)$,
%
we say that $(a_1,\dots , a_k)$ \emph{realizes} the table $\beta$
%
%
%
and write $tb_{\mathfrak{A}}(a_1,\dots ,a_k) = \beta$.
%
Every tuple of $k$ distinct elements in the $\eta$-structure $\mathfrak{A}$
%
realize exactly one $k$-table $\beta$ over $\eta$.
\end{comment}
%

%
For investigations on two-variable logic, we also need the notion of a $2$-type.
Recalling that we let $x$ and $y$ denote, respectively, $v_1$ and $v_2$ in 
two-variable contexts, we define that a \emph{$2$-type} over $\eta$ is a conjunction $\beta(x,y)
\wedge\alpha_1(x)\wedge \alpha_2(y)\wedge x\not=y$, where $\beta$ is a $2$-table
while $\alpha_1$ and $\alpha_2$ are $1$-types over $\eta$.
Such a $2$-type can be conveniently 
denoted by $\alpha_1\beta\alpha_2$.
Let $\gamma$ be either a $1$-type or a $k$-table over $\eta$. Let $L_+$ and $L_-$ be the sets of 
positive and negative literals in $\gamma$.
Given weight functions $w:\eta\rightarrow \mathbb{Q}$ and $\bar{w}:\eta\rightarrow \mathbb{Q}$,
the \emph{weight of\, $\gamma$}, denoted by $\langle w,\bar{w}\rangle(\gamma)$, is the product
%
%
%
%\begin{equation}
%
$\prod\limits_{R\overline{v} \, \in\, L_+}w(R)\ \ \ \ 
\boldsymbol{\cdot} \ \ \ \
\prod\limits_{\neg R\overline{v} \, \in\, L_-}\bar{w}(R),$
%
%\end{equation}
%
%\smallskip
%
\noindent
where $\overline{v}$ denotes all the different possible tuples of 
variables in the literals of $\gamma$.
\subsection{A Skolemization procedure}
\label{sect:skolemn}

We now define a formula transformation procedure designed for the purposes of
model counting. The procedure, which was originally
introduced in 
\cite{DBLP:conf/kr/BroeckMD14}, resembles Skolemization
but does not in general produce an equisatisfiable formula.
Here we present a slightly modified variant of the
procedure from \cite{DBLP:conf/kr/BroeckMD14}
suitable for our purposes.

If $Q\in\{\exists, \forall\}$ is a quantifier, we let $Q'$
denote the \emph{dual quantifier} of $Q$, i.e., $Q' \in \{\exists,
\forall\}\setminus\{Q\}$. Let

\smallskip

$\empty{ }\ \ \ \ \ \ \ \
\ \ \ \ \ \ \ \ \varphi\, :=\, \forall x_1\dots \forall x_k\exists y_1\dots
\exists y_m Q_1z_1\dots Q_n z_n\, \psi$

\smallskip 

\noindent
be a first-order prenex normal 
form sentence where $\psi$ is quantifier-free and $Q_i\in\{\exists,\forall\}$ for all $i$.
%with $Q_1 = \forall$.
%
%
%
We eliminate the block $\exists y_1\dots \exists y_m$ of
existential quantifiers of $\varphi$ in two steps. First we replace $\varphi$ by
%This is done in two steps.
%First we replace $\varphi$ by the sentence
%

\smallskip

$\text{ }\ \ \ \ \ \ \forall x_1\dots \forall x_k(Ax_1\dots x_k \vee \neg
\exists y_1\dots \exists y_m Q_1z_1\dots Q_n z_n\, \psi ),$

\smallskip

\noindent
where $A$ is a fresh $k$-ary predicate. Then the negation is pushed inwards
past the quantifier block $\exists y_1\dots \exists y_m Q_1z_1\dots Q_n z_n$ and the resulting dual
block $\forall y_1\dots \forall y_m Q_1'z_1\dots Q_n' z_n$ is pulled out so
that we end up with the prenex normal form sentence

\smallskip

$\text{ }\ \ \ \ \ \ \ \ \forall x_1\dots \forall x_k\forall y_1\dots
\forall y_m Q_1'z_1\dots Q_n' z_n\, (Ax_1\dots x_k \vee \neg\psi).$

\smallskip

%Let $\varphi$ be a prenex normal form sentence.
Let $\mathit{Sk}_0(\varphi)$ denote the sentence obtained by changing
the maximally long outermost block of existential quantifiers (the
block $\exists y_1\dots \exists y_m$ if $Q_1 = \forall$ above) to a block of
universal quantifiers using the above two steps,
and let $\mathit{Sk}(\varphi)$ be the $\forall^*$-sentence
obtained by repeatedly applying $\mathit{Sk}_0$.
%
%
%
%
%
%
%
%
%
%
%By repeating this procedure, we
%end up with an $\forall^*$-sentence in prenex normal form.
%For any prenex normal form sentence $\varphi$, we let $\mathit{Sk}(\varphi)$ denote
%the $\forall^*$-sentence obtained as described above.
For any conjunction $\chi := \psi_1\wedge\dots \wedge \psi_n$ of prenex normal form
sentences, we let $\mathit{Sk}(\chi) :=
\mathit{Sk}(\psi_1)\wedge\dots \wedge \mathit{Sk}(\psi_n)$.
The next Lemma is proved similarly as 
the corresponding result in \cite{DBLP:conf/kr/BroeckMD14}.
For the sake of completeness, Appendix \ref{skolemlemmaproof} also gives a
\nolinebreak proof.
\begin{lemma}[cf. \cite{DBLP:conf/kr/BroeckMD14}]\label{skolemlemma}
Let $\chi$ and $\varphi$ be sentences, $\varphi$ a 
conjunction of prenex
normal form sentences. Let $w$ and $\bar{w}$
be weight functions. Then
%
%for $\varphi\wedge\chi$. Then
%
%
%
%\begin{enumerate}
%
%\item
%$\varphi^*$ is a finite conjunction of $\mathrm{U}_1$-formulae
%of type $\forall x_1\dots\forall x_k\, \psi$, where $\psi$ is quantifier-free.
%(Note that $k$ need not be the same for all conjuncts of $\varphi^*$.)
%
%\item

\smallskip

$\text{ }\ \ \ \mathrm{WFOMC}(\varphi\wedge\chi,n,w,\bar{w})
= \mathrm{WFOMC}(\mathit{Sk}(\varphi)\wedge\chi,n,w',\bar{w}'),$

\smallskip

\noindent
where $w'$ and $\bar{w}'$ are obtained from $w$ and $\bar{w}$ by 
mapping the fresh symbols in $\mathit{Sk}(\varphi)$ to
$1$ in the case of\, $w'$ and to $-1$ in the case of\, $\bar{w}'$.
If $\varphi$ is a sentence of\, $\mathrm{FO}^2$, then so is $\mathit{Sk}(\varphi)$.
If $\varphi$ is a sentence of $\mathcal{L}
\in\{\mathrm{SU}_1,\mathrm{U}_1\}$ in generalized Scott normal form,
then $\mathit{Sk}(\varphi)\in\mathcal{L}$.
%
%
%
%\footnote{\color{blue}shouldn't $w'$ map them to 1 instead of -1? ... Yes}

%\end{enumerate}
%
%
%
\end{lemma}
%

\begin{comment}
We note that \cite{DBLP:conf/kr/BroeckMD14}
shows how to define a Skolemization operator, let us
call it $\mathit{Sk}'$, that turns only the single
outermost existential to a universal one so 
that for example $\mathit{Sk'}(\forall x \exists y \forall z \psi)
= \forall x \forall y \forall z \psi'$ becomes
\end{comment}

\begin{comment}
We defined the operator $\mathit{Sk}$ so that it turns all of the 
outermost existential quantifier block in one step to a 
block of universal quantifiers. 
step and iterates such elimination steps until we have an $\forall^*$-sentence.
It is possible to define an operator $\mathit{Sk}'$ that eliminates only one
existential quantifier of a prefix (the outermost one),
\end{comment}

\subsection{Further syntactic assumptions}\label{syntaxsection}

%
\begin{comment}
By Lemmas \ref{scottlemma2} and \ref{skolemlemma},
for proving that $\mathrm{SU}_1$-sentences have 
polynomial time weighted model counting (Theorem \ref{maintech}), it
suffices to consider $\mathrm{SU}_1$-sentences of the 
type $\chi = \forall x_1\chi_1\wedge \dots \wedge 
\forall x_1\dots \forall x_k\chi_k$ each $\chi_i$
where each $\chi_i$ is quantifier-free. 
%

\begin{comment}
For proving the corresponding result for 
sentences of $\mathrm{FO}^2$
with a functionality axiom, we may restrict attention to 
sentences $\forall x \forall y \varphi\, \wedge\, \forall x \exists y^{=1}\psi(x,y)$
where $\varphi$ and $\psi(x,y)$ are quantifier-free. To see this,
note that a formula $\chi\, \wedge\, \forall x \exists^{=1}\chi'(x,y)$,
where $\chi$ is an $\mathrm{FO}^2$-sentence $\chi'$
and $\chi'$ a $\mathrm{FO}^2$-sentence, can be brough to the 
desired form by applying the
\end{comment}

Let $\varphi$ be a sentence of $\mathrm{U}_1$.
Due to Lemmas \ref{scottlemma2} and \ref{skolemlemma}, we have
$$\mathrm{WFOMC}(\varphi,n,w,\bar{w}) 
= \mathrm{WFOMC}(\mathit{Sk}(\mathit{Sc}(\varphi)),n,w',\bar{w}'),$$
where $w'$ and $\bar{w}'$ treat the fresh symbols as
discussed when defining $\mathit{Sc}$ and $\mathit{Sk}$.
Call $\chi := \mathit{Sk}(\mathit{Sc}(\varphi))$
and assume, w.l.o.g., that $\chi = \forall x_1\chi_1\wedge\dots \wedge 
\forall x_1\dots \forall x_k\chi_k$ for some
matrices $\chi_i$. For technical convenience,
when working with $\mathrm{SU}_1$, we assume
that there is at most one $\forall^*$-conjunct of any particular
width; if not, formulae $\forall x_1\dots \forall x_p\chi'$ and $\forall x_1\dots \forall x_p\chi''$
can always be combined to $\forall x_1\dots \forall x_p(\chi'\wedge\chi'')$.

Now, $\chi$ may contain nullary predicates.
Let $S$ be the set of nullary predicates of $\chi$ and 
let $f:S\rightarrow \{\top,\bot\}$ be a function. 
Let $\chi^f$ be the formula obtained from $\chi$ by 
replacing each nullary predicate $P$ by $f(P)$.
It is easy to compute $\mathrm{WFOMC}(\chi,n,v,\bar{v})$
from the values $\mathrm{WFOMC}(\chi^f,n,v,\bar{v})$
for all functions $f:S\rightarrow \{\top, \bot\}$. Thus, when
studying $\mathrm{WFOCM}$ for $\mathrm{U}_1$
and $\mathrm{SU}_1$, we begin with a formula $\forall x_1\chi_1\wedge\dots \wedge 
\forall x_1\dots \forall x_k\chi_k$ assumed to be
free of nullary predicates. We also assume, w.l.o.g., that 
\begin{comment}
We also assume that this sentence
contains at least one higher arity relation; if not, first ensure 
that there is a $\forall^*$-conjunct $\forall x_1\forall x_2\chi_2$ in
the formula (adding a dummy conjunct if necessary) and then replace
$\chi_2$ by $\chi_2\wedge R(x_1,x_2)$ and
put $w(R) = \bar{w}(R) = 1$.
\end{comment}
the greatest width $k$ is at least $2$ and equal to the greatest arity of
relation symbols occurring in the formula. (We can 
always add dummy $\forall^*$-conjuncts of higher width, 
and we can add a dummy $k$-ary symbol $R$ to a
conjunct $\forall x_1\dots \forall x_k\chi_k$ by replacing $\chi_k$
by $Rx_1\dots   x_k\wedge \chi_k$ and setting $w(R) = \bar{w}(R) = 1$.)

We then turn to two-variable logic with a
functionality axiom. Consider a sentence $\varphi' := \varphi\wedge\forall x\exists^{=1}y\psi(x,y)$,
where $\varphi$ and $\psi(x,y)$ are $\mathrm{FO}^2$-formulae.
By applying the Scott normal form
procedure for eliminating quantified subformulae and
using the Skolemization operator $\mathit{Sk}$, it is easy to 
obtain (see Appendix \ref{appendixmillion}) a sentence 
$\varphi'' := \forall x\forall y\chi \wedge \forall x \exists^{=1} y \chi'(x,y)$
with $\chi$ and $\chi'(x,y)$ quantifier-free so that 
$\mathrm{WFOMC}(\varphi',n,w,\bar{w})
= \mathrm{WFOMC}(\varphi'',n,w',\bar{w}')$, where $w'$ and $\bar{w}'$
extend $w$ and $\bar{w}$. If $\varphi''$ has nullary predicates, we
eliminate them in the way discussed above. Thus, when studying
$\mathrm{WFOMC}$ for $\mathrm{FO}^2$ with a functionality axioms below, we
begin with a sentence of the form $\forall x \forall y \varphi_1 \wedge
\forall x \exists y^{=1} \varphi_2(x,y)$ where $\varphi_1$ and $\varphi_2$ are
quantifier-free. We also assume, w.l.o.g., that the sentence
contains at least one binary relation symbol and no symbols of
arity greater than two. (These assumptions are easy to justify,
see Appendix \ref{higherarityappendix}.)

\begin{comment}
For later technical simplicity, if 
$\chi$ contains $\forall^*$-sentences of width one (with only one
universal quantifier) as conjuncts, we modify
these conjuncts to equivalent $\forall^*$-sentences of
width two. The resulting sentence $\chi'$ is of the form
$\forall x_1\forall x_2\chi_2\wedge\dots \wedge 
\forall x_1\dots \forall x_k\chi_k$.
\end{comment}

%%%%%%%%%%%%%%%%%
%%%%%%%%%%%%%%%%%
%%%%%%%%%%%%%%%%%

\newcommand{\vp}{\ensuremath{\varphi}}
\newcommand{\ybf}{\ensuremath{\mathbf y}}
\newcommand{\Smc}{\ensuremath{\mathcal S}}
\newcommand{\mn}[1]{\ensuremath{\mathsf{#1}}}

%

%HERE

%
\section{Counting for $\mathrm{FO}^2$ with functionality}\label{funcounting}

We now show that the symmetric weighted model counting problem for
$\mathrm{FO}^2$-sentences with a functionality axiom is in $\mathrm{PTIME}$. As
discussed in the preliminaries, it suffices to consider a formula

\smallskip

$\text{ }\ \ \ \ \ \ \ \ \ \ \ \ \ \ \ \ 
\Phi_0\, :=\, \forall x\forall y\, \varphi_0^{\forall}(x,y)\ \wedge\ \forall x\exists^{=1} y\, \varphi_0^{\exists}(x,y),$

\smallskip

\noindent
where $\varphi_0^{\forall}(x,y)$ and $\varphi_0^{\exists}(x,y)$ are
quantifier-free and do not contain nullary relation symbols. Further assumptions 
justified in the preliminaries are that $\Phi_0$ contains at least one binary relation
symbol and no relation symbols of arity greater than two. From now on, we
thus consider a fixed formula $\Phi_0$ of the above form as well as
fixed weight functions $w$ and $\bar w$.
%  and devote
% the rest of Section~\ref{funcounting} to studying the weighted model
% counting problem for $\Phi_0$.
% 
%

To simplify the constructions below, it would help if the subformula
$\varphi_0^{\exists}(x,y)$ of $\Phi_0$ was of the form $x\not= y\
\wedge\, \psi$ %for some formula $\varphi'$
so that a witness for the existential quantifier 
%(i.e., an interpretation of the
%variable $y$ of $\varphi_0^{\exists}$) 
would always be different from
the point it is a witness to.
% (i.e., the interpretation of $x$).
However, there seems to be no obvious way to convert $\Phi_0$ into 
the desired form while preserving weighted model counts. We
thus use a conversion that does not preserve these counts 
and then show how to rectify this. Let
\[
\begin{array}{l}
\Phi\ :=\ \forall x\forall y\,     \bigl(  \varphi_0^{\forall}(x,y)  \ 
\wedge\   \neg  (x\not= y \wedge \varphi_0^{\exists}(x,x)\wedge\varphi_0^{\exists}(x,y))\, \bigr)
\end{array}
\]
\[
\begin{array}{lllll}
%
%\Phi\ :=\ \forall x\forall y\,   &  \bigl( & \varphi_0^{\forall}(x,y) & \ 
%
%\wedge\  & \neg  (x\not= y \wedge \varphi_0^{\exists}(x,x)\wedge\varphi_0^{\exists}(x,y)\bigr)\\ \\
%
& \text{ } \wedge\ \forall x\exists^{=1} y\, \bigl(\, 
x\not= y 
               & \wedge\ \bigl(  &  &(\varphi_0^{\exists}(x,x) \wedge Sy)\\
 &             &      \empty\ \    & \vee & (  \varphi_0^{\exists}(x,x) \wedge Sx \wedge Ty)\\
 &             &      \empty\ \    & \vee & (  \neg\varphi_0^{\exists}(x,x) \wedge \varphi_0^{\exists}(x,y))\ \bigr)\bigr),
\end{array}
\]
where $S$ and $T$ are fresh unary predicates. Let $\mathcal{M}$ be the class
of models (over $\mathit{voc}(\Phi)$\hspace{0.0mm}) where $S$ and $T$ are
interpreted to be distinct singletons. Slightly abusing notation,
assume further that both $w$ and $\bar w$ assign to both $S$ and $T$ the value $1$.

The remainder of this section is devoted to showing how to compute
$$\mathrm{WFOMC}(\Phi,n,w,\bar{w})\upharpoonright{\mathcal{M}}.$$
%the variation of $\mathrm{WFOMC}(\Phi,n,w,\bar{w})$ that only considers models from
%$\mathcal{M}$.
We note that the class

\smallskip

$\text{ }\ \ \ \ \ \ \ \ \ \ \ \ \ \ \ \ \ \ \ 
\mathcal{M}_1 := \{\, \mathfrak{M}\in\mathcal{M}\, |\, 
\mathit{dom}(\mathfrak{M}) = n\, \}$

\smallskip

\noindent
of models relevant to $\mathrm{WFOMC}(\Phi,n,w,\bar{w})\upharpoonright {\mathcal M}$ can
be obtained from the class
$\mathcal{M}_0$ of models
relevant to $\mathrm{WFOMC}(\Phi_0,n,w,\bar{w})$ by interpreting $S$ and $T$ as
distinct singletons in all possible ways, so every model
in $\mathcal{M}_0$ gives rise to $n(n-1)$ models in
$\mathcal{M}_1$. It is thus easy to see that we get
$\mathrm{WFOMC}(\Phi_0,n,w,\bar{w})$ from $\mathrm{WFOMC}(\Phi,n,w,\bar{w})
\upharpoonright {\mathcal M}$ by dividing by $n(n-1)$.
(The case $n=1$ is computed separately.)
We note that there seems to be no obvious way to modify $\Phi$ to additionally
enforce $S$ and $T$ to be distinct singletons. While this property is
expressible by a sentence of $\mathrm{FO}^2$, adding such a sentence 
would destroy the intended syntactic structure of $\Phi$. Note here that
Lemma \ref{skolemlemma} does not in general produce an equivalent formula, so
using it for modifying the required $\mathrm{FO}^2$-sentence would not help.

\subsection{Partitioning models}\label{partitioningmodels}

For simplicity, let $\Phi =\forall x\forall y\,
\varphi^{\forall}(x,y)\ \wedge \,
%\bigwedge\limits_{i\, \in\, [m_{\exists}] }
\forall x\exists^{=1} y\, \varphi^\exists(x,y)$, so $\varphi^{\forall}(x,y)$ and $\varphi^\exists(x,y)$
denote, respectively, the quantifier-free parts of the $\forall\forall$-conjunct
and $\forall\exists^{=1}$-conjunct of $\Phi$.
%the quantifier-free part of the $\forall\exists^{=1}$-conjunct of $\Phi$ is denoted
%simply by $\varphi(x,y)$.
%
%
%
% We note that neither $\varphi^{\forall}$ nor $\varphi$ contains
% nullary predicates and that $\varphi$ is of the form $x\not=y\, \wedge\, \varphi'$
% for some $\varphi'$.
In the rest of Section \ref{funcounting}, types and tables mean types and tables with
respect to $\mathit{voc}(\Phi)$.
Now, recall from the preliminaries that a $2$-type $\tau(x,y)$ is a conjunction
$\alpha(x)\wedge\beta(x,y)\wedge\alpha'(y)\wedge x\not=y$
where $\beta$ is a $2$-table and $\alpha,\alpha'$ are $1$-types.  We
denote such a $2$-type by $\alpha\beta\alpha'$.  We call
$\alpha$ the \emph{first $1$-type} and $\alpha'$ the
\emph{second $1$-type} of $\tau(x,y)$ and denote
these $1$-types by $\tau(1)$ and $\tau(2)$.
The 2-type $\tau(x,y)$ is \emph{coherent} if

\smallskip

$\text{ }\ \ \ \ \ \ \ \ \ 
\tau(x,y)\models\varphi^{\forall}(x,y)\wedge\varphi^{\forall}(y,x)\wedge\varphi^{\forall}(x,x)
\wedge\varphi^{\forall}(y,y).$

\smallskip

\noindent
A $1$-type $\alpha(x)$ is \emph{coherent} if
$\alpha(x)\models\varphi^{\forall}(x,x)$.
The \emph{inverse} of a $2$-type $\tau(x,y)$ is the $2$-type
$\tau'(x,y)\equiv \tau(y,x)$. A $2$-type is \emph{symmetric} if it is
equal to its inverse.  

The \emph{witness} of an element $u$ in a model $\mathfrak{M}$ of $\Phi$
is the unique element $v$ such that
$\mathfrak{M}\models\varphi^\exists(u,v)$.  A $2$-type $\tau(x,y)$ is
\emph{witnessing} % , or a \emph{witness 2-type},
if $\tau(x,y)$ is coherent and we have $\tau(x,y)\models
\varphi^\exists(x,y)$. % So the
% term `witness' refers to witnesses for the existential quantifier in
% $\Phi$.l
%A $2$-type $\tau(x,y)$ is 
The $2$-type $\tau(x,y)$ is \emph{both ways witnessing} if both
it and its inverse are witnessing; note that a both ways witnessing
$2$-type can be symmetric but does not have to. The set of all 
witnessing $2$-types is denoted by $\Lambda$.
%
%
%
%We fix a linear order, denoted simply by $<$, over the finite set of all $2$-types,
%and denote its reflexive variant with $\leq$. 
%We let $\Gamma$ denote the set of all witness 2-types.  

\newcommand{\Mmf}{\mathfrak{M}} We next define the notions of a
\emph{block} and a \emph{cell}. These are an essential part of the
subsequent constructions. One central idea of our model counting strategy
is to partition the domain of a model $\mathfrak{M}$ of $\Phi$ into
blocks which are further partitioned into cells. A
\emph{block type} is simply a witnessing 2-type. The \emph{block type
  of an element} $u$ of $\mathfrak{M}\models\Phi$ is the unique
witnessing $2$-type $\tau(x,y)$ such that $\mathfrak{M}\models
\tau(u,v)$, where $v$ is the witness of $u$.  The domain $M$ of
$\mathfrak{M}$ is partitioned by the family $(B^\Mmf_{\tau})_{\tau}$
where each set $B^\Mmf_\tau\subseteq M$ contains precisely the
elements of $\mathfrak{M}$ with block type $\tau$. Some of the
sets $B^\Mmf_{\tau}$ can of course be empty. We call the sets
$B^\Mmf_{\tau}$ the \emph{blocks} of $\mathfrak{M}$ and refer to
$B^\Mmf_{\tau}$ as the \emph{block of type} $\tau$.  We fix a linear
order~$<$ over all block types and denote its reflexive variant by
$\leq$.

%
% The set $M_{\sigma}$
% is referred to as the \emph{block of type $\sigma$} (in
% $\mathfrak{M}$).
%
%
%
%
%
%
%

%
Each block further partitions into \emph{cells}.
A \emph{cell type} is a pair $(\sigma,\tau)$ of witnessing 2-types.
For brevity, we denote cell types by $\sigma\tau$ instead of
$(\sigma,\tau)$. % for simplicity.
The \emph{cell type of an element $u$} in a model
$\mathfrak{M}\models\Phi$ is the unique pair $\sigma\tau$ such that
$u \in B_\sigma^\Mmf$ and $v \in B_\tau^\Mmf$, $v$ the witness of $u$.
Each block $B_\sigma^\Mmf$ is partitioned by the family
$(C^\Mmf_{\sigma\tau})_{\tau}$ where each set $C^\Mmf_{\sigma\tau}\subseteq
B_\sigma^\Mmf$ contains precisely the elements of $\mathfrak{M}$ that
are of cell type $\sigma\tau$.  Again, some of the sets
$C^\Mmf_{\sigma\tau}$ can be empty. We call the sets
$C^\Mmf_{\sigma\tau}$ the \emph{cells} of $B^\Mmf_\sigma$ and refer to
each $C^\Mmf_{\sigma\tau}$ as the \emph{cell of type} $\sigma\tau$.
%We fix a linear oder $<$ over all cell types and denote its reflexive
%variant by $\leq$.\footnote{\color{blue}This does not seem to be used
%anywhere.}

% The set of cell types is denoted by $\mathcal{C}$. We fix a linear order over $\mathcal{C}$
% and denote it by $<$ and $\leq$. We note that blocks are indexed by
% witness types $\sigma\in \Gamma$, 
% and therefore blocks can be ordered accordingly by the order of $\Gamma$.
% Thus we write $B_{\sigma}\leq B_{\tau}$ if $\sigma\leq \tau$.
% %

%
%We denote the set of all cells by $\mathcal{C}$ and 
%the set of blocks by $\mathcal{B}$. We define a linear order over $\mathcal{C}$
%and also on $\mathcal{B}$. Again we denote these order simply by $<$ and $\leq$.
%

%
\subsection{The counting strategy}\label{countingstrategysection}

\newcommand{\Cmc}{\mathcal{C}}
\newcommand{\Mmc}{\mathcal{M}}

We now describe our strategy for
computing $\mathrm{WFOMC}(\Phi_0,n,w,\bar{w})$ informally.
A formal treatment will be given later on. We 
first explain how to compute $\mathrm{WFOMC}(\Phi,n,w,\bar{w})$
and then discuss 
how to get $\mathrm{WFOMC}(\Phi,n,w,\bar{w})
\upharpoonright {\mathcal M}$ and $\mathrm{WFOMC}(\Phi_0,n,w,\bar{w})$.
%
\begin{comment}
\smallskip

%
$\text{ }\ \ \ \ \ \ \hspace{1cm}
\ \ \ \ \ \ \mathrm{WFOMC}(\Phi,n,w,\bar{w})\upharpoonright {\mathcal M}$
%

\smallskip
\end{comment}

%
\noindent

The strategy for computing $\mathrm{WFOMC}(\Phi,n,w,\bar{w})$ is
based on blocks and cells.
We are interested
in models of a given size $n$ and with domain $n = \{0,\dots , n-1\}$, so we
let $\Mmc^\Phi_n$ denote the set of all $\mathit{voc}(\Phi)$-models $\mathfrak{M}$
with domain $n$ that satisfy $\Phi$.
%
%$\Mmc^\Phi_n := \{\, \mathfrak{M}\in\mathcal{M}\, |\,
%\mathfrak{M}\models\Phi,\ \mathit{dom}(\mathfrak{M}) = n\, \}$.
%denote the
%class of all models of $\Phi$ from $\Mmc$ that have domain $n$.
% Thus all models discussed
% below are assumed to have the domain $n$.
%

%
A \emph{cell configuration} is a partition
$(C_{\sigma\tau})_{\sigma\tau}$ of the set $n$ where some sets can be
empty. The \emph{cell configuration of a model} $\mathfrak{M} \in
\Mmc^\Phi_n$ is the family $(C^\Mmf_{\sigma\tau})_{\sigma\tau}$ as
defined in Section \ref{partitioningmodels}. For a cell configuration
$\Gamma$, we use $\Mmc^\Phi_{n,\Gamma}$ to denote the class of all
models in $\mathcal{M}^\Phi_n$ that have cell configuration
$\Gamma$. It is clear that the family
$(\Mmc^\Phi_{n,\Gamma})_{\Gamma}$, where $\Gamma$ ranges over all cell
configurations, partitions $\mathcal{M}^\Phi_n$ (though some sets $\Mmc^\Phi_{n,\Gamma}$
can be empty). It would
be convenient to iterate over cell configurations $\Gamma$ and
independently compute the weight of all models in each~$\Mmc^\Phi_{n,\Gamma}$,
eventually summing up the computed weights. However, this option is
ruled out since the number of cell configurations is exponential in
$n$.  Fortunately, it suffices to only know the \emph{sizes}
of cells rather than their concrete extensions.
Let $\sigma_1,\dots,\sigma_k$ enumerate all block types.  Then
the sequence $$\sigma_1\sigma_1,\sigma_1\sigma_2,\dots,\sigma_k\sigma_k$$ enumerates
all cell types. A \emph{multiplicity configuration} is a vector
$$\scalebox{0.99}{$(n_{\sigma_1\sigma_1},n_{\sigma_1\sigma_2},\dots,n_{\sigma_k\sigma_k})$}$$
where each $n_{\sigma_i\sigma_j}$ is a number in $\{0,\dots , n\}$ and
$\scalebox{0.99}{$n_{\sigma_1\sigma_1} + \dots + n_{\sigma_k\sigma_k} = n$}$. The
\emph{multiplicity configuration of a model} $\mathfrak{M} \in
\Mmc^\Phi_n$ is obtained by letting each $n_{\sigma\tau}$ be the size
of $C^\Mmf_{\sigma\tau}$.  For a multiplicity configuration $\Delta$,
we use $\Mmc^{\Phi}_{n,\Delta}$ to denote the class of all models from
$\Mmc^\Phi_n$ that have multiplicity configuration $\Delta$. Clearly,
the number of multiplicity configurations
%the number of vectors $(n_{\sigma_1\sigma_1},
%n_{\sigma_1\sigma_2},\dots,n_{\sigma_k\sigma_k})$ such that 
%$n_{\sigma_1\sigma_1} + \dots + n_{\sigma_k\sigma_k} = n$
is polynomial in $n$, so we can iterate over them and---as we shall see---independently
compute the weight of all models in each~$\Mmc^\Phi_{n,\Delta}$ in polynomial time.
Each cell configuration gives rise to a unique multiplicity
configuration. Conversely, for every multiplicity configuration $\Delta=
\scalebox{0.99}{$(n_{\sigma_1\sigma_1},
n_{\sigma_1\sigma_2},\dots,n_{\sigma_k\sigma_k})$}$,
there are
$${\ell=\binom{n}{n_{\sigma_1\sigma_1},n_{\sigma_1\sigma_2}, \dots ,
  n_{\sigma_k\sigma_k}}}
$$
cell configurations giving rise
to~$\Delta$.  For any two such cell configurations $\Gamma$,
$\Gamma'$, the
weight of $\mathcal{M}^\Phi_{n,\Gamma}$ (i.e., the sum of the
weights of the models in $\mathcal{M}^\Phi_{n,\Gamma}$) is identical to the
weight of $\mathcal{M}^\Phi_{n,\Gamma'}$. To obtain the weight of
$\mathcal{M}^\Phi_{n,\Delta}$, it thus suffices to consider a
single cell configuration $\Gamma$ giving rise to $\Delta$,
compute the weight of $\mathcal{M}^\Phi_{n,\Gamma}$ and multiply by $\ell$. 
We now briefly describe how to compute the number of models in 
$\mathcal{M}^\Phi_{n,\Gamma}$\hspace{0.4mm}, ignoring weights.
With easy modifications, the approach will ultimately also give the
weight of $\mathcal{M}^\Phi_{n,\Gamma}$.
%(which means the 
%sum of the weights of the models in $\mathcal{M}^\Phi_{n,\Gamma}$).
%
Although our algorithm is not going to explicitly construct
the models in~$\mathcal{M}^\Phi_{n,\Gamma}$, to describe
how the number of those models is counted, we
simultaneously consider how we could construct all of them.
Let $(B_\sigma)_\sigma$ be the \emph{block configuration} that 
corresponds to the cell
configuration $\Gamma=(C_{\sigma\tau})_{\sigma\tau}$, that is, $B_\sigma = \bigcup_\tau C_{\sigma\tau}$ for each block type $\sigma$. 
As the domain is fixed to be $n$, we consider all
possible ways to assign $1$-types to the elements of $n$ and $2$-tables to
pairs of distinct elements such that we realize the cell
configuration~$\Gamma$. There is no freedom for the
$1$-types: if $u \in B_{\sigma}$, then we must 
assign the $1$-type $\sigma(1)$ to $u$.
%
% We of course assign each set $C_{\sigma\tau}$
% the $1$-type $\sigma(1)$, thereby fixing all $1$-types over the set $n$.
% %There is only one way of doing this once the protocells are indeed fully fixed.
% With (only) the $1$-types fixed, we still need to
% count the weighted number of
% appropriately assigning each pair of distinct elements in $n$ a $2$-table.
% %more correctly, count the number of ways to
% %assign $2$-tables. We next describe
% %how this is done.
% %
% %
% %
% %
% This is done by considering \emph{blocks}. A block of type $\sigma$ is of
% course the union $B_{\sigma}$ of the sets $C_{\sigma\tau}$ for all $\tau$.
% We note that strictly speaking, the set $B_{\sigma}$ becomes a block of type $\sigma$ 
% once a model has been fully constructed, not necessarily any sooner, and similarly, $C_{\sigma\tau}$ becomes a 
% cell of type $\sigma\tau$ at the end of the process. Still, we shall speak of $B_{\sigma}$ and $C_{\sigma\tau}$ as
% the block of type $\sigma$ and cell of type $\sigma\tau$.
% %
%
%
%
To assign $2$-tables, we consider each pair of blocks
$(B_{\sigma},B_{\tau})$ with $\sigma\leq \tau$ independently,
identifying each possible way to simultaneously assign $2$-tables to
pairs in $B_{\sigma} \times B_{\tau}$. (When $\sigma=\tau$, we must be careful to (1) 
consider only pairs $(u,v)$ of \emph{distinct} elements and (2) to
assign a $2$-table to only one of $(u,v),(v,u)$.)
It is important to understand that
in~$B_{\sigma}$, there is exactly one cell, namely $C_{\sigma\tau}$,
whose elements require a witness from $B_{\tau}$.
Similarly, in $B_{\tau}$, it is precisely the elements of $C_{\tau\sigma}$ that
require a witness in~$B_{\sigma}$. Since witnesses are unique, we
start with identifying the ways to simultaneously define
functions $f:C_{\sigma\tau}\rightarrow B_{\tau}$ and $g:C_{\tau\sigma}\rightarrow B_{\sigma}$
that determine the witnesses. It then remains to count the number of ways to
assign $2$-types to the remaining edges that are witnessing in neither
direction. This is easy---as long as we know the number $N$ of
these remaining edges---since each edge realizes the $1$-type $\sigma(1)$ at the one
end and $\tau(2)$ at the other. We use a look-up
table to find the number of $2$-tables that are `compatible' with this.
%We use brute force to compute the
%number of $2$-tables that are `compatible' with this.
The number $N$
depends on how many pairs in $B_{\sigma}\times B_{\tau}$ and $B_{\tau}\times B_{\sigma}$ 
belong to the functions that
determine the witnesses, but $N$ will nevertheless be easy to determine, as we shall see.
The precise arithmetic formulae for counting the number of ways to
assign $2$-tables to all elements from $B_{\sigma} \times B_{\tau}$ are
given in Section~\ref{combinatorics}. There are several cases that
need to be distinguished. We now briefly look at the 
most important cases informally.
We start with the case $\sigma=\tau$, that is, the two blocks
$B_{\sigma}, B_{\tau}$ are in fact the same single block, and we aim to
assign $2$-tables within that block. Then exactly the elements from
the cell $C_{\sigma\sigma}$ require a witness in $B_{\sigma}$
itself. If $\sigma$ is not both ways witnessing, then
$C_{\sigma\sigma}$ will be the domain of an anti-involutive function
$C_{\sigma\sigma}\rightarrow B_{\sigma}$ that determines a witness in
$B_{\sigma}$ for each element in $C_{\sigma\sigma}$.  If $\sigma$ is
both ways witnessing and its own inverse, this function is
involutive. The case where $\sigma$ is both ways witnessing but not
its own inverse is pathological in the sense that there are then no
valid ways to assign $2$-tables unless $C_{\sigma\sigma}$ is empty. To sum up, in
each case, the core
task in designing the desired arithmetic formula is thus to count the
number of suitable anti-involutive or involutive functions.
%  and adding a (comparably
% simple) formula that counts the number of ways to assign $2$-tables to
% the pairs in $(B_\sigma \setminus C_{\sigma\sigma}) \times (B_\sigma
% \setminus C_{\sigma\sigma})$.
%
% possible, but then---unless
% $C_{\sigma\sigma}$ is empty---the number of ways to connect
% $B_{\sigma}$ to itself is zero, as an element $u$ supporting $\sigma$
% cannot have another element $v$ that also supports $\sigma$ as its
% witness, because $v$ would then have $u$ as its witness via the
% inverse of $\sigma$ (which is not $\sigma$).
%

Now consider the case where $\sigma\not=\tau$ and thus $B_{\sigma}$
and $B_{\tau}$ are different blocks. Here again several subcases arise
based on whether $\sigma$ and $\tau$ are both ways witnessing. The most
interesting case is where neither $\sigma$ nor $\tau$ is both ways
witnessing. We then need to count the ways of finding two
functions $f:C_{\sigma\tau}\rightarrow B_{\tau}$ and
$g:C_{\tau\sigma}\rightarrow B_{\sigma}$ that are nowhere inverses of
each other.  In the case where $\sigma$ and $\tau$ are both
ways witnessing and inverses of each other, we need to count the
number of perfect matchings between the sets $C_{\sigma\tau}$ and
$C_{\tau\sigma}$.  The case where at least one of the witness types,
say $\sigma$, is both ways witnessing, but $\sigma$ and $\tau$ are not
inverses of each other, is again pathological.
Implementing the above ideas, we will show how to obtain, for any pair of blocks
$B_\sigma,B_\tau$, where we have $\sigma\leq\tau$, a function
%
%
%
%\begin{equation}\label{temporaryequation}
$M_{\sigma\tau}(n_{\sigma},n_{\sigma\tau},n_{\tau},n_{\tau\sigma})$
%,\overline{w})
%\end{equation}
%
%
%
that counts the `weighted number of ways' to connect the blocks
$B_{\sigma}$ and $B_{\tau}$ with $2$-tables, when given the sizes
$n_{\sigma}$ and $n_{\tau}$ of the blocks as well as the sizes
$n_{\sigma\tau}$ and $n_{\tau\sigma}$ of the cells
$C_{\sigma\tau}\subseteq B_{\sigma}$ and $C_{\tau\sigma}\subseteq
B_{\tau}$; we note that while this fixes the intuitive
interpretation of $M_{\sigma\tau}(n_{\sigma},n_{\sigma\tau},n_{\tau},n_{\tau\sigma})$, the function $M_{\sigma\tau}$
will become formally defined in terms of arithmetic operations in Section \ref{finalsn}. 
(Furthermore, for the sake of extra clarity,
we provide in Appendix \ref{1234990appendix} a more detailed 
description of what the weighted number of ways to
connect $B_{\sigma}$ and $B_{\tau}$ with $2$-tables means.)

Recall that $\Lambda$ is the set of all block types and note that
%
%
%
%\begin{equation}\label{exponequt}
$n_{\sigma} = \sum\limits_{\sigma'\, \in\, \Lambda} n_{\sigma\sigma'}
$
%\end{equation}
%
%
%
and likewise for $n_{\tau}$, so $n_{\sigma}$ and $n_\tau$ are determined by the
sizes of all cells in the blocks $B_{\sigma}$ and $B_{\tau}$.
With the aim of achieving notational uniformity, we can thus replace
$M_{\sigma\tau}$ by a function
\begin{equation}\label{nomoretemporaryequation}
N_{\sigma\tau}(n_{\sigma_1\sigma_1},n_{\sigma_1\sigma_2},\dots , n_{\sigma_k\sigma_k})
%,\overline{w})
\end{equation}
that outputs
$M_{\sigma\tau}(n_{\sigma},n_{\sigma\tau},n_{\tau},n_{\tau\sigma})$
but has a full multiplicity type as an input.
%kkkkkkkkkk
% Letting $n_{\sigma}$ to denote the right hand side of
% Equation \ref{exponequt}, w
Noting that the weight functions $w$ and $\bar w$ give rise to the
weight $w_\alpha := \langle w,\bar{w}\rangle(\alpha)$
of each $1$-type $\alpha$,
we now observe that we can compute $\mathrm{WFOMC}(\Phi,n,w,\bar{w})$ by the function
\begin{multline}\label{ultimatecountingalgzero}
\mathcal{U}(n)\ :=
\sum\limits_{n_{\sigma_1\sigma_1}+n_{\sigma_1\sigma_2} + \dots +n_{\sigma_k\sigma_k} =\, n}
\Bigl(\ \ 
\binom{n}{n_{\sigma_1\sigma_1},n_{\sigma_1\sigma_2}, \dots , n_{\sigma_k\sigma_k}}\\
\boldsymbol{\cdot}\bigl(\prod\limits_{
%\substack{
%\eta\, \in\, \mathcal{C}\\
\sigma\, \in\ \Lambda
%\mu\nu\, \in\, \mathcal{D}_0}
%
}
(w_{\sigma(1)})^{n_{\sigma}}\bigr)
\prod\limits_{
%\substack{
%\eta\, \in\, \mathcal{C}\\
\sigma,\, \tau\ \in\ \Lambda
%\mu\nu\, \in\, \mathcal{D}_0}
%
}N_{\sigma\tau}(n_{\sigma_1\sigma_1},n_{\sigma_1\sigma_2},
\dots , n_{\sigma_k\sigma_k})\, \Bigr).
\end{multline}
%
%
%
%where $n_{\sigma\tau}$ represents the number of elements in the cell $C_{\sigma\tau}$.
%
%{\color{blue}Do we want to explain the formula in a few words?} Risking not to, as
%as a consequence of the space limitation.
%

%
Recall, however, that we aim to compute
$\mathrm{WFOMC}(\Phi,n,w,\bar{w})\upharpoonright \Mmc$ rather than
$\mathrm{WFOMC}(\Phi,n,w,\bar{w})$. And eventually we want to compute
$$\mathrm{WFOMC}(\Phi_0,n,w,\bar{w}),$$ which can be obtained simply by dividing
$\mathrm{WFOMC}(\Phi,n,w,\bar{w})\upharpoonright \Mmc$ by $n(n-1)$.  In order to 
get from $\mathrm{WFOMC}(\Phi,n,w,\bar{w})$ to
$\mathrm{WFOMC}(\Phi,n,w,\bar{w})\upharpoonright \Mmc$, we need to discard weights
contributed by models where $S$ and $T$ are not interpreted as
non-overlapping singletons. This is easy: we only need to discard multiplicity
configurations $(n_{\sigma_1\sigma_1},n_{\sigma_1\sigma_2},\dots ,
n_{\sigma_k\sigma_k})$ that do not make $S$ and $T$ distinct
singletons. Let $\langle{n}\rangle$ be the set of multiplicity configurations
with the undesired ones excluded.
Summing up, $\mathrm{WFOMC}(
\Phi_0,n,w,\bar{w})$ can thus be computed by the function
\begin{multline}\label{ultimatecountingalg}
\mathcal{W}(n)\ =
\dfrac{1}{n(n-1)}\\
\boldsymbol{\cdot}\sum\limits_{(n_{\sigma_1\sigma_1},n_{\sigma_1\sigma_2},
\dots,n_{\sigma_k\sigma_k})\, \in\, \langle n\rangle}
\Bigl(\ \binom{n}{n_{\sigma_1\sigma_1},n_{\sigma_1\sigma_2}, \dots , n_{\sigma_k\sigma_k}}\\
\boldsymbol{\cdot}\bigl(\prod\limits_{
%\substack{
%\eta\, \in\, \mathcal{C}\\
\sigma\, \in\ \Lambda
%\mu\nu\, \in\, \mathcal{D}_0}
%
}
(w_{\sigma(1)})^{n_{\sigma}}\bigr)
\prod\limits_{
%\substack{
%\eta\, \in\, \mathcal{C}\\
\sigma,\, \tau\ \in\ \Lambda
%\mu\nu\, \in\, \mathcal{D}_0}
%
}N_{\sigma\tau}(n_{\sigma_1\sigma_1},n_{\sigma_1\sigma_2},\dots , n_{\sigma_k\sigma_k}) \Bigr).
\end{multline}
%
%
%
%where $n_{\sigma\tau}$ represents the number of elements in the cell $C_{\sigma\tau}$.
%
%
%
%
%
%
%

%
%
%
In the next Section~\ref{combinatorics} we deal with the 
combinatorics for defining the 
functions $N_{\sigma\tau}$. The actual functions $N_{\sigma\tau}$ are then
specified in Section~\ref{finalsn} where we conclude our argument.
\subsection{The relevant combinatorics}\label{combinatorics}
Let $k\in\mathbb{N}$.
The following equation is well known.
\begin{equation}\label{alternatingsum}
\sum\limits_{i=0}^{i=k}(-1)^i\binom{k}{i} = 
\begin{cases}
0\text{ if }k\not= 0\\
1\text{ if }k= 0.
\end{cases}
\end{equation}
On the intuitive level, the \emph{alternating sum} on the left hand side of the
equation relates directly to the \emph{inclusion-exclusion} principle.
We shall make
frequent use of this equation in the constructions below.

%
%$$(n-1)^m - (n-1)^{m-2}\binom{m}{2} + (n-1)^{m-4}\binom{m}{4}\frac{(2i)!}{2^i}
%-\dots * (n-1)^{m-m}\binom{m}{m}\frac{(2i)!}{2^i}$$
%
%
%

%
The first result of this section, Proposition \ref{nonidempcounting1} below,
will ultimately help us in counting the number of ways to
connect a block \emph{to itself} with $2$-tables. However, the result is interesting in its 
own right and thus we formulate it abstractly, like most 
results in this section, without
reference to $2$-types or other logic-related notions.
%
%The same principle applies to all results in this section: we shall 
%formulate each of them without logic-related notions.
%

%
Recall that a unary function is \emph{anti-involutive} if $f(f(x))\not= x$ 
for all $x\in \mathit{dom}(f)$. Note that this implies $f(x) \not= x$
for all $x\in \mathit{dom}(f)$, i.e., $f$ is \emph{fixed point free}.
\begin{proposition}\label{nonidempcounting1}
Let $n$ and $m\leq n$ be nonnegative integers.
The number of anti-involutive functions $m\rightarrow n$ is
\begin{equation}\label{alternatingsumnoninvolutive1}
\mathit{I}(m,n) := \sum\limits_{i = 0}^{i = \lfloor m/2
\rfloor}(-1)^i(n-1)^{m-2i}\binom{m}{2i}\frac{(2i)!}{2^i(i!)}.
\end{equation}
\end{proposition}
\begin{proof}
We first note that for a nonnegative integer $i$,
there are $\binom{2i}{2,\dots , 2}\frac{1}{i!}$ ways to partition $2i$ elements 
into doubletons, where $2$ is written $i$ times in
the bottom row. Writing the multinomial coefficient $\binom{2i}{2,\dots , 2}$ open, we 
see that $\binom{2i}{2,\dots , 2}\frac{1}{i!}=\frac{(2i)!}{2^i(i!)}$.
%Thus there are $\frac{(2i)!}{2^i}$
%ways to partition a set of size $2i$ to doubletons.
%

%
Now, for a fixed point free function $f$, if $f(f(x)) = x$ for some $x$, then we call
the doubleton $\{x,f(x)\}$ a \emph{symmetric pair} of $f$.
A \emph{fixed point free function $f:m\rightarrow n$ with $i$ labelled 
symmetric pairs} is a pair $(f,L)$ where $f:m\rightarrow n$ is a fixed point
free function and $L$ is a set of exactly $i$ symmetric pairs of $f$. Note that $f$ may
have other symmetric pairs outside $L$, so $L$ only 
distinguishes $i$ specially labelled symmetric pairs.
It is easy to see that the number of fixed point free functions $m\rightarrow n$ with $i$
labelled symmetric pairs is given by 
\begin{equation}\label{idempzerostep1}
(n-1)^{m-2i}\binom{m}{2i}\frac{(2i)!}{2^i(i!)}.
\end{equation}
Therefore Equation \ref{alternatingsumnoninvolutive1} has the following intuitive
interpretation. The equation first counts---when $i$ is zero---all 
fixed point free functions $m\rightarrow n$ without any \emph{labelled} symmetric pairs;
unlabelled symmetric pairs are allowed.
Then, when $i=1$, the equation subtracts
the number of fixed point free functions $m\rightarrow n$ with one
\emph{labelled} symmetric pair. Then, with $i=2$ the equation adds the the number of 
fixed point free functions $m\rightarrow n$ with two \emph{labelled} symmetric pairs,
and so on, all the way to $i= \lfloor m/2 \rfloor$.
Now, fix a single fixed point free function $f:m\rightarrow n$ with 
\emph{exactly} $j$ symmetric pairs. Labelling $k\leq j$ of the $j$ 
symmetric pairs can be done in $\binom{j}{k}$ ways.
Thus $f$ gets counted in Equation \ref{alternatingsumnoninvolutive1} precisely
%
%
%
%\begin{equation}\label{localwhateverquuquukuu}
$S(j) := \binom{j}{0} -  \binom{j}{1} + \binom{j}{2} - \dots * \binom{j}{j}$
%\end{equation}
%
%
%
times, where $*$ is $+$ if $j$ is even and $-$ if $j$ is odd.
By Equation \ref{alternatingsum}, $S(j)$ is $0$ when $j\not= 0$ and $1$ when $j=0$.
Thus $f$ gets counted zero times if $j\not=0$
and once if $j=0$.
%
%
%
% once in the sum.
%
%
%
\end{proof}
Proposition \ref{nonidempcounting1} will be used for  
counting functions that find a witness for 
each element of a cell $C$ of size $m$ from a block $B\supseteq C$ of
size $n$. However, we also need to count the ways of assigning
non-witnessing $2$-tables to the remaining edges inside $B$.
The next two results, Lemma \ref{nonskolemlemma}
and Proposition \ref{internalcounting1}, will help in this.
%We first provide some auxiliary definitions.
%

%
Let $G$ be an undirected graph with the
set $V$ of vertices and $E$ of edges. A \emph{labelling} of $G$ 
with $k$ \emph{symmetric colours} and $\ell$
\emph{directed colours} is a pair of functions $(s,d)$
such that 
\begin{enumerate}
\item
$s$ maps some set $U\subseteq E$ into $[k]$, not
necessarily surjectively,
\item
$d$ maps the complement $E\setminus U$ of $U$ into $[\ell]\times V$
such that each edge $e\in E\setminus U$ gets mapped to a pair $(i,u)$
where $u\in e$. Intuitively, $d$ picks a colour in $[\ell]$ and an 
\emph{orientation} for $e$. It is \emph{not} required that each $i\in [\ell]$
gets assigned to some edge.
\end{enumerate}
The colour $j\in [\ell]$ is said to\ \emph{define a function}\ \ if the relation
$\{\, (u,v)\, |\, \{u,v\}\in E\setminus U,\ d(\{u,v\}) = (j,v)\, \}$
is a function.
Rather than counting labellings of graphs, we need to count \emph{weighted
labellings}: a weighted labelling of a graph $G$ with $k$ symmetric and $\ell$ directed
colours is a triple

\smallskip

$\text{ }\ \ \ \ \ \ \ \ \ \ \ \ \ \ \ \ \ \ \ \ 
W = ((s,d),(w_1,\dots, w_k),(x_1,\dots x_{\ell}))$

\smallskip

\noindent
such that $(s,d)$ is a labelling of $G$
and $w_1,\dots , w_k$ are weights of the symmetric colours $1,\dots , k$ and $x_1,\dots , x_{\ell}$ 
weights of the directed colours $1,\dots , \ell$. (Here e.g. $1$ is called both a 
directed and symmetric colour. This will pose no problem.) The \emph{total weight} $t_W$ of
the weighted labelling $W$ is the 
product of the weights assigned to the edges of $G$.
The \emph{weighted number of
labellings} of $G$ with $k$ symmetric and $\ell$ directed 
colours with weights $w_1,\dots , w_k$ and $x_1,\dots , x_{\ell}$ is the sum of the total weights $t_W$ of all
weighted labellings $W = ((s,d),(w_1,\dots, w_k),(x_1,\dots x_{\ell}))$ of $G$.
The following is easy to prove (see Appendix \ref{aabbccdd}).
\begin{lemma}\label{nonskolemlemma}
The function 
\begin{multline}\label{whatevertherestfunctionweight}
L_{\empty_{k,\ell}}(N,w_1,\dots , w_k,x_1,\dots ,x_{\ell}) :=\\
\sum\limits_{
\substack{i_1 + \dots + i_k + j_1 + \dots + j_{\ell}\, =\, N}}
\Bigl(\ \binom{N}{i_1,\dots , i_k,j_1,\dots j_{\ell}}\\
\cdot 2^{j_1+\dots + j_{\ell}}
\bigl(\prod\limits_{\substack{p\, \in\, [k]}}(w_p)^{i_p}\bigr)
\bigl(\prod\limits_{\substack{q\, \in\, [\ell ] }}(x_q)^{j_q}\bigr)\ \Bigr)
\end{multline}
gives the weighted number of labellings of an arbitrary $N$-edge graph  
with $k$ symmetric and $\ell$ directed colours 
with weights $w_1,\dots , w_k$ and $x_1,\dots , x_{\ell}$.
\textcolor{black}{At least one of $k,\ell$ is assumed nonzero here.
The first (resp. second) product on the bottom row outputs $1$ if $k=0$
(resp. $\ell = 0$).}
\end{lemma}
\textcolor{black}{
{We also define $L_{0,0}(N) := 0$ for $N > 0$
and $L_{0,0}(0) := 1$, and furthermore, $L_{k,\ell}(m,
w_1,\dots , w_k,x_1,\dots ,x_{\ell}) := 0$ for all negative 
integers $m$.}} The following is
easy to prove (see Appendix \ref{aaabbccdd}).
\begin{proposition}\label{internalcounting1}
Let $n$ and $m\leq n$ be nonnegative integers,
and let $w_1,\dots , w_k$ and $x_1,\dots , x_{\ell},y$ be
weights for $k$ symmetric and $\ell + 1$ directed colours. The function 
\begin{multline}\label{twoconfunctionone}
J_{\empty_{k,\ell+1}}(m,n,w_1,\dots , w_k,x_1,\dots, x_l,y)\, :=\\
I(m,n) \cdot y^{m} \cdot L_{\empty_{k,\ell}}\bigl(\binom{n}{2}-m,\, w_1,\dots , w_k,x_1,\dots ,x_{\ell}\bigr)
%\cdot\sum\limits_{i + j\, =\, \binom{n}{2}-m}
%\binom{\binom{n}{2}-m}{i}k^{i}\cdot (2\ell)^{j}
%
\end{multline}
gives the weighted number of labellings of the complete $n$-element graph
with $k$ symmetric and $\ell+1$ directed colours
with the above weights 
such that the edges of colour $\ell+1$ define an
anti-involutive function $m\rightarrow n$.
\end{proposition}
%
%
%
%\begin{proof}
%
%The relatively easy proof is given in Appendix \ref{aaabbccdd}.
%
%\end{proof}
%

%
\begin{comment}
Let $A$ and $B$ be disjoint sets, and let $A_0\subseteq A$ and $B_0\subseteq B$.
Two functions $f:A_0\rightarrow B$ and $g:B_0\rightarrow A$ are said to be
\emph{nowhere inverses} of each other if we have $g(f(u))\not= u$ and $f(g(v))\not= v$
for all $u\in A_0$ and $v\in B_0$.
\end{comment}
%The \emph{complete bipartite graph on $A\times B$} is the
%undirected complete bipartite graph with the set $\{\, \{a,b\}\, |\, a\in A,\ b\in B\, \}$ of edges.
The following result will ultimately help us in counting 
the ways of connecting two \emph{different} blocks to each other with $2$-tables.
%
%The result is also interesting in its own right.
%

%
\begin{proposition}\label{nonidempcounting}
Let $A\not=\emptyset$ and $B\not=\emptyset$ be
disjoint finite sets, $|A| = M$ and $|B| = N$.
Let $A_m\subseteq A$ and $B_n\subseteq B$ be 
sets of sizes $m$ and $n$, respectively. There exist
\begin{equation}\label{bipartitecccccc}
{K(m,M,n,N) :=
\sum\limits_{i = 0}^{\textcolor{black}{i\, =\, \mathit{min}(m,n)}}(-1)^{i}
\binom{m}{i}\binom{n}{i}\bigl(i! \cdot M^{(n - i)}\cdot N^{(m - i)}\bigr)}
\end{equation}
ways to define two functions $f:A_m\rightarrow B$ and $g:B_n\rightarrow A$
that are nowhere inverses of each other.
\end{proposition}
\begin{proof}
Fix some $i\leq \mathit{min}(m,n)$, and fix two sets $A_i\subseteq A_m$ and $B_i\subseteq B_n$,
both of size $i$. There exist $\bigl(i! \cdot M^{(n - i)}\cdot N^{(m - i)}\bigr)$
ways to define a pair of functions $f: A_m\rightarrow B$ and $g:B_n\rightarrow A$
such that $f\upharpoonright A_i$ and $g\upharpoonright B_i$
are bijections and inverses of each other; here $i!$ is the number of 
ways the two functions can be defined in restriction to $A_i$ and $B_i$ so
that they become inverses of each other over $A_i$ and $B_i$. 
(Note that $f$ and $g$ can be inverses elsewhere too.) Thus 
$$\binom{m}{i}\binom{n}{i}\bigl(i! \cdot M^{(n - i)}\cdot N^{(m - i)}\bigr)$$
gives the number of tuples $(f,g,A',B')$ 
such that $f:A_m\rightarrow B$ and $g:B_n\rightarrow A$ are functions 
and $A'\subseteq A_m$ and $B'\subseteq B_n$ sets of size $i$
such that $f\upharpoonright A'$ and $g\upharpoonright B'$ are 
inverses of each other.
Now, fix two sets $A_j\subseteq A_m$ and $B_j\subseteq B_n$ of
size $j$ both. Fix two functions $f:A_m\rightarrow B$ and $g:B_n\rightarrow A$ that are inverses of 
each other on $A_j$ and $B_j$ and \emph{nowhere else}.
Thus the pair $f,g$ is counted in the alternating sum of
Equation \ref{bipartitecccccc} exactly 
%
%
%
%\begin{equation}\label{againgodknowswhat}
$S(j) := \binom{j}{0} - \binom{j}{1} + \binom{j}{2} - \dots * \binom{j}{j}$
%\end{equation}
%
%
%
times, where $*$ is $+$ if $j$ is even and $-$ otherwise.
By Equation \ref{alternatingsum}, $S(j)$ is zero when $j\not= 0$ and one when $j= 0$.
Thus the pair $f,g$ gets counted zero times if $j\not=0$ and otherwise once.
\end{proof}
\textcolor{black}{We also define $K(m,M,n,N) := 0$
for any $m\leq M$ and $n\leq N$ with $M=0\not=n$ or $N=0\not=m$.
Furthermore, we define $K(0,0,0,N)  = K(0,M,0,0)  = 1$
for all $M,N\in\mathbb{N}$.}
The next result, Proposition \ref{bipartitecounting1}, extends
Proposition \ref{nonidempcounting} so
that also the non-witnessing edges will be taken into account. To formulate 
the result, we define that for disjoint finite sets $A$ and $B$,
the \emph{complete bipartite graph on $A\times B$} is the
undirected bipartite graph with the set $\{\, \{a,b\}\, |\, a\in A,\ b\in B\, \}$ of edges.
%
%
%
%If $M$ and $N$ are nonnegative integers, then the \emph{complete $M\times N$ bipartite graph} is
%the complete bipartite graph on $A\times B$ for some disjointe sets $A$ and $B$ of 
%sizes $M$ and $N$, respectively. 
%

%
\begin{proposition}\label{bipartitecounting1}
Let $A$ and $B$ be
finite disjoint sets, $|A| = M$ and $|B| = N$.
Let $A_m\subseteq A$ and $B_n\subseteq B$ be 
sets of sizes $m$ and $n$, respectively.
Let $w_1,\dots , w_k$ and $x_1,\dots , x_{\ell},y,z$ be weights.
The function  
\begin{multline}\label{twoconfunction2}
\scalebox{0.91}[1]{$P_{\empty_{k,\ell+2}}(m,M,n,N,w_1,\dots , w_k,x_1,\dots ,x_{\ell}, y, z) :=$}\\
\scalebox{0.91}[1]{$K(m,M,n,N) \cdot {y}^m z^n \cdot L_{\empty_{k, \ell}}(MN - m - n,\ w_1,
\dots , w_k,x_1,\dots ,x_{\ell})$}
\end{multline}
gives the weighted number of labellings of the complete bipartite graph on $A\times B$  
with $k$ symmetric and $\ell+2$ directed colours 
with weights $w_1,\dots , w_k$ and $x_1,\dots , x_{\ell},y,z$
such that the directed colours $\ell +1$ and $\ell + 2$
%(with weights $y$ and $z$)
define, respectively, functions $f: A_m\rightarrow B$ and $g: B_n\rightarrow A$
that are nowhere inverses of each other.
\end{proposition}
\begin{proof}
The relatively easy proof is given in Appendix \ref{JLX}.
\end{proof}
The results so far in this section provide us with ways of counting in
cases where witnesses are found via $2$-types that are not both ways witnessing.
We now deal with the remaining cases.
Recall that $n!!$ denotes the standard \emph{double factorial} operation defined
such that for example $7!! = 7\cdot 5 \cdot 3 \cdot 1$ and $8!! = 8\cdot 6 \cdot 4 \cdot 2$.
We define the function $F:\mathbb{N}\rightarrow\mathbb{N}$ 
such that $F(0) = 1$ and for all $m\in \mathbb{Z}_+$, we have
$F(m) = (m-1)!!$ if $m$ is even and $F(m) = 0$ otherwise. It is well known and 
easy to show that $F(m)$ is precisely the number of \emph{perfect matchings} of the 
complete graph $G$ with the set $m$ of vertices, i.e., the number of \emph{$1$-factors} of a
graph of order $m$ (and with the set $m$ of vertices). By a perfect matching of
the set $m$, we refer to a perfect matching of the complete graph with the vertex set $m$.
The following is easy to prove (see Appendix \ref{gghhlljj}).
\begin{proposition}\label{internalcounting2}
Let $n$ and $m\leq n$ be nonnegative integers, and
let $w_1,\dots , w_k,y$ and $x_1,\dots , x_{\ell}$ be weights. The function
\begin{multline}\label{dunnohowmaniethequation}
S_{\empty_{k+1,\ell}}(m,n,w_1,\dots , w_k,y,x_1,\dots ,x_{\ell})\, :=\\ F(m)\cdot y^{m/2}
\cdot L_{\empty_{k,\ell}}\bigl(\binom{n}{2} - \lfloor m/2 \rfloor,\, w_1,\dots , w_k,x_1,
\dots ,x_{\ell}\bigr)
%\cdot\sum\limits_{i + j\, =\, \binom{n}{2}-m}
%\binom{\binom{n}{2}-m}{i}k^{i}\cdot (2\ell)^{j}
%
\end{multline}
gives the weighted number of labellings of the complete graph
with the set $n$ of vertices
with $k+1$ symmetric and $\ell$ directed colours
with weights $w_1,\dots , w_k,y$ and $x_1,\dots , x_{\ell}$
such that the edges of the symmetric colour $k+1$ define a
perfect matching of the set $m\subseteq n$.
\end{proposition}
Let $F':\mathbb{N}\times\mathbb{N}\rightarrow\mathbb{N}$ be the function
such that $F'(n,m) = n!$ if $n = m$ and $F'(n,m) = 0$ otherwise.
A perfect matching between two disjoint sets $S$ and $T$ is a 
perfect matching of the complete bipartite graph on $S\times T$.
The following is immediate.
\begin{proposition}\label{bipartitecounting22}
Let $A$ and $B$ be
finite disjoint finite sets, $|A| = M$ and $|B| = N$.
Let $A_m\subseteq A$ and $B_n\subseteq B$ be 
sets of sizes $m$ and $n$, respectively.
%Let $w_1,\dots , w_k,y$ and $x_1,\dots , x_{\ell}$ be weights.
The function
\begin{multline}\label{twoconfunction22}
%
%\sum\limits_{\substack{\ell_1+\dots +\ell_k\\ =\, mn - m - n}}\binom{mn - m - n}{\ell_1,\dots , \ell_k}
%
%
%
T_{\empty_{k+1,\ell}}(m,M,n,N,w_1,\dots , w_k,y,x_1,\dots ,x_{\ell}) :=\\
F'(n,m) \cdot {y}^n \cdot L_{\empty_{k, \ell}}(MN - n,\ w_1,\dots , w_k,x_1,\dots ,x_{\ell})
\end{multline}
gives the weighted number of labellings of the complete bipartite graph on $A\times B$ 
with $k+1$ symmetric and $\ell$ directed colours 
with weights $w_1,\dots , w_k,y$ and $x_1,\dots , x_{\ell}$
such that the symmetric colour $k+1$ defines a 
perfect matching between $A_m$ and $B_n$.
\end{proposition}
%
%
%
%\begin{proof}
%
%Straightforward.
%
%\end{proof}
%

%
\subsection{Defining the functions $N_{\sigma\tau}$}\label{finalsn}

We now discuss how the functions $N_{\sigma\tau}$ are defined
for all pairs $\sigma\tau$ of block types, thereby completing the 
definition of Equation \ref{ultimatecountingalg}.

%
\begin{comment}
The formula for weighted model counting of the sentence $\Phi_0$ will be
given by Equation \ref{ultimatecountingalg} once the functions $N_{\sigma\tau}$ in
that equation are defined for all pairs $\sigma\tau$ of block types.
We now discuss how the individual functions $N_{\sigma\tau}$ are defined.
\end{comment}
%

%
Fix a pair $\sigma\tau$ of block types.
%
%
%
%Let $\sigma(1)$ and $\tau(1)$ denote, respectively,
%the first $1$-types of $\sigma$ and $\tau$.
%
Let $y$ and $z$, respectively, be the
weights of the $2$-tables of the $2$-types $\sigma$ and $\tau$. Let $w_1,\dots , w_k$
(respectively, $x_1,\dots , x_{\ell}$) enumerate the weights of the
symmetric (resp., unsymmetric) $2$-tables $\beta$ that can connect the block $B_{\sigma}$ to
the block $B_{\tau}$ so that neither the resulting $2$-type $\sigma(1)\beta\tau(1)$ 
nor its inverse is witnessing, and furthermore,  $\sigma(1)\beta\tau(1)$ is
coherent. If $\sigma = \tau$, these are the
weights of the coherent $2$-tables that can connect a
point in block $B_{\sigma}$ to another point in the same block so that the 
resulting $2$-type is not witnessing in either direction.
We next consider different cases depending on how $\sigma$ and $\tau$ relate to each other.
We let $\overline{n}$
denote the input tuple to $N_{\sigma\tau}$
with $\overline{n}$ containing the multiplicities $n_{\sigma'\sigma''}$ of 
all cell types $\sigma'\sigma''$.
%and $\overline{w}$ the
%weights for all $1$-types and $2$-tables.
%(We note that $\bar{w}$ and $\overline{w}$
%are not the same.)
For a witness $2$-type $\sigma'$,  we let $n_{\sigma'}$ abbreviate
%
%(cf. Equation \ref{exponequt})
the sum
$\sum_{\sigma''\, \in\, \Lambda}\ \ n_{\sigma'\sigma''}$  (recall $\Lambda$ is
the set of all block types).
The witness $2$-type $\sigma'$ is \emph{compatible} with a
witness $2$-type $\sigma''$ if $\sigma'(2) = \sigma''(1)$.
% i.e., the first $1$-type of $\sigma'$ is the same as the second $1$-type of $\sigma''$.
%

\medskip

\noindent
\textbf{Case 1.}
We assume that \textbf{1.a)} $\sigma\not=\tau$;
\textbf{1.b)} $\sigma$ and $\tau$ are
\emph{compatible} with each other; \textbf{1.c)} 
neither $\sigma$ nor $\tau$ is a
both ways witnessing $2$-type.
By
%Equation \ref{twoconfunction2} in
Proposition \ref{bipartitecounting1},
the weight contributed by all the
edges from $B_{\sigma}$ to $B_{\tau}$ is thus given by
%
%
%
%\begin{multline}\label{needtoinventtoomanynames}

\smallskip 

\scalebox{1}[1]{$
N_{\sigma\tau}(\overline{n}) :=
P_{\empty_{k,\ell+2}}(n_{\sigma\tau},n_{\sigma}, n_{\tau\sigma},
n_{\tau},w_1, ... , w_k,x_1, ... ,x_{\ell}, y, z).$}
%\end{multline}
%
%
%

\smallskip

\noindent
which defines $N_{\sigma\tau}$ under these
particular assumptions.
The remaining cases are similar but use different functions
defined in the previous section. For example, when $\sigma = \tau$ 
and $\sigma$ is not two-ways witnessing, we 
use the function $J_{\ell,k+1}$ from Equation \ref{twoconfunctionone}
in Proposition \ref{internalcounting1}; see the Appendix
\ref{definingfunctionsappendix} (Case 4) for the full
details. All the remaining cases are also discussed in Appendix \ref{definingfunctionsappendix}.
By inspecting the operations of Equation \ref{ultimatecountingalg}, we
conclude the following.

\begin{theorem}
The weighted model counting problem of each two-variable logic sentence with a
functionality axiom is in $\mathrm{PTIME}$.
\end{theorem}
%

%%%%%%%%%%%%%%%%
%%%%%%%%%%%%%%%%
%%%%%%%%%%%%%%%%

%%%%%%%%%%%%%%%%
%%%%%%%%%%%%%%%%
%%%%%%%%%%%%%%%%

%%%%%%%%%%%%%%%%
%%%%%%%%%%%%%%%%
%%%%%%%%%%%%%%%%

\section{Weighted model counting for $\mathrm{U}_1$}

In this section we prove that $\mathrm{WFOMC}$  is in $\mathrm{PTIME}$
for each sentence of $\mathrm{U}_1$.
To that end, we first establish the same result for $\mathrm{SU}_1$,
stated as Lemma \ref{maintech} below.
%The main result here is Lemma \ref{maintech}.
We follow a proof strategy that makes explicit how the 
syntactic restrictions of $\mathrm{SU}_1$ naturally lead to polynomial time 
model counting. We then provide a reduction from $\mathrm{U}_1$ to $\mathrm{SU}_1$.
%The main result for $\mathrm{U}_1$ is Theorem~\ref{reductionuone}
%in Section~\ref{thiswasthe}.
%
%
%
\begin{comment}
The aim of the section is not only to get Theorem \ref{maintech} proved, but also to
provide a proof that explains how the syntactic
restrictions of $\mathrm{SU}_1$ lead to that result.
\end{comment}
%

%
\subsection{Weighted model counting for $\mathrm{SU}_1$}\label{fullyuniformsection}
%

%\subsection{Formula packing}

%We begin by providing some auxiliary definitions and results.
%We begin with an auxiliary definition,
%and after that we define a crucial syntactic transformation.
%
%
%

Let $\psi(x_1,\dots, x_k)$ be a quantifier-free first-order formula, and
let $\ell \leq k$ be a positive integer.
Let $F$ denote the set of all surjections $[k]\rightarrow [\ell]$.
The conjunction $\bigwedge\{\, \psi(x_{f(1)},\dots,x_{f(k)})\ |\ f\in F\, \}$
is called the \emph{$\ell$-surjective image} of $\psi$.
%If $\ell = k$, the conjunction is
%called the \emph{symmetric image} of $\psi$.
%

%
\begin{definition}
Let $\varphi$ be a conjunction of $\forall^*$-sentences of $\mathrm{FO}$
(These need not be sentences of $\mathrm{U}_1$ or $\mathrm{SU}_1$.)
We now define the \emph{surjective completion} $\mathit{sur}(\varphi)$ of $\varphi$
by modifying $\varphi$ as follows.

\medskip

\noindent
\textbf{1.)}
%Let $k_a$ be the maximum arity of relation symbols in $\varphi$ and $k_w$
%the maximum width of the $\forall^*$-conjuncts of $\varphi$. Call $k :=\mathit{max}(k_a,k_w)$.
Let $k$ be the maximum width of the $\forall^*$-conjuncts of $\varphi$.
We modify $\varphi$ so that for all $i\in [k]$, there exists a conjunct of 
width $i$. This can be ensured by adding dummy conjuncts, if necessary.
We let $\varphi'$ denote the resulting sentence.

\smallskip

\noindent
\textbf{2.)}
We merge the conjuncts of $\varphi'$ with the
same width, so that for example $\forall x \forall y \psi(x,y)
\wedge \forall x \forall y \chi(x,y)$ would become $\forall x \forall y(\psi(x,y)\wedge\chi(x,y))$.
Thus the resulting formula $\varphi''$ is a
conjunction of $\forall^*$-sentences so that no two conjuncts have the same width.

\smallskip

\noindent
\textbf{3.)}
%Let $\chi$ be the symmetric image of $\psi_k(x_1,\dots , x_k)$ and 
%Replace the conjunct $\forall v_1\dots\forall v_k\,
%\psi_k(v_1,\dots,v_k)$ in $\varphi'$ by $\forall v_1\dots\forall v_k\, \chi$.
Define $\varphi_k'' := \varphi''$ where $k$ is the maximum 
width of the $\forall^*$-sentences of $\varphi''$.
%
%
%
%
% \smallskip
%
% \noindent
% \textbf{4.)}
Inductively, let $1\leq \ell < k$ and assume we
have defined a sentence
$\varphi_{\ell+ 1}''\  =\ 
\chi_1\wedge\dots \wedge \chi_k$
where each $\chi_i$ is an $\forall^*$-sentence of
width $i$. Let $\psi_{\ell + 1}$ and $\psi_{\ell}$ 
be the matrices of $\chi_{\ell + 1}$ and $\chi_{\ell}$,
so we have\\

\vspace{-3mm}

$\text{ }$\hspace{1.30cm} $\chi_{\ell+1} = \forall x_1\dots\forall x_{\ell+1}\, \psi_{\ell+1}(x_1,\dots,x_{\ell+1}),$\\
$\text{ }$\hspace{1.63cm} $\chi_{\ell}\ \ \, \, \,  = \forall x_1\dots\forall x_{\ell}\,
\psi_{\ell}(x_1,\dots,x_{\ell}).$\\

\vspace{-3mm}

\noindent
Let $\psi_{\ell}'$ denote the $\ell$-surjective image of $\psi_{\ell + 1}$.
%and $\psi_{\ell}''$ the symmetric image of $\psi_{\ell}\wedge \psi_{\ell}'$.
Replace the conjunct $\chi_{\ell}$ of $\varphi_{\ell+1}''$ by
$\forall x_1\dots\forall x_{\ell}( \psi_{\ell}\wedge \psi_{\ell}')$.
Define $\varphi_{\ell}''$ to be the resulting
modification of $\varphi_{\ell + 1}''$.
Define $\mathit{sur}(\varphi)$ to be the formula $\varphi_1''$.
%
%
%
%
%
%
%
%\noindent
%\textbf{5.)} We define $\mathit{sur}(\varphi)$ to be the formula $\varphi_1''$.
%
\end{definition}
%

%
\begin{comment}
\begin{example}
%
The above process transforms  $\forall x\forall y\forall z Sxyz
\wedge \forall x\forall y Rxy$ to the following formula
(where we have removed repeated conjuncts):
%
%
%
\begin{multline*}
%
\forall x\forall y \forall z( Sxyz\wedge Sxzy \wedge Syxz \wedge Syzx
%
\wedge Szxy \wedge Szyx)\\
%
\forall x\forall y( Rxy\wedge Ryx \wedge Sxyx\wedge Sxxy
%
\wedge Syxx \wedge Syyx
%
\wedge Syxy \wedge Sxyy)
%
%
%
\end{multline*}
%
%
%
\end{example}
\end{comment}
%

%
%It is easy to see that the surjective
%completion of a first-order $\forall^*$-sentence $\varphi$ is 
%equivalent to $\varphi$.
%
%
%
Let $\varphi := \forall x_1\dots \forall x_k\psi$ be an $\forall^*$-sentence.
We let $\mathit{diff}(\varphi)$ denote the sentence
$\forall x_1\dots\forall x_k(\mathit{diff}(x_1,\dots, x_k) \rightarrow \psi)$, letting
$\mathit{diff}(x_1) := \top$.
For a conjunction $\varphi' := \varphi_1\wedge\dots \wedge\varphi_k$ of $\forall^*$-sentences, we
define $\mathit{diff}(\varphi')
:=\mathit{diff}(\varphi_1)\wedge\dots \wedge\mathit{diff}(\varphi_k)$.
\begin{lemma}\label{difsurlemma}
We have $\varphi\equiv \mathit{diff}(\mathit{sur}(\varphi))$
for any conjunction $\varphi$ of first-order\, $\forall^*$-sentences.
\end{lemma}
\begin{proof}
Clearly $\varphi\equiv\mathit{sur}(\varphi)$. 
Also $\mathit{sur}(\varphi)\equiv\mathit{diff}(\mathit{sur}(\varphi))$, as $\mathit{sur}$ is
based on steps where the surjective image of a matrix is
pushed to be part of the matrix of a formula with one variable less.
\end{proof}
As discussed in the
preliminaries, to prove that the weighted model counting
problem of $\mathrm{SU}_1$-sentences is in $\mathrm{PTIME}$, it
suffices to show this for conjunctions of $\forall^*$-sentences of $\mathrm{SU}_1$ of the type
%
%
%
%\begin{equation}\label{varphi10equation}
$\varphi'\ =\ 
\forall x_1\, \psi_1'\ \wedge\ 
%\forall x_1\forall x_2\forall x_3\, \psi_3'\ \wedge
\dots \wedge\ \forall x_1 \dots \forall x_{p}\, \psi_p'$
%\end{equation}
%
%
%
where each $\psi_i'$ is quantifier-free.
Other assumptions justified in the preliminaries
are that $\varphi'$ contains no nullary atoms; $p$ is equal to
the greatest arity of the symbols in $\mathit{voc}(\varphi')$; and $p\geq 2$.
%and $\varphi'$
%contains at least one relation symbol of arity greater than two.
By Lemma \ref{difsurlemma}, $\varphi'$ is 
equivalent to $\varphi'' := \mathit{diff}(\mathit{sur}(\varphi'))$.
We remove the conjunct of width 1 from $\varphi''$ and integrate it to
the conjunct of width 2, so if
\begin{center}
$\varphi'' =
\forall x_1\, \chi_1(x_1)\wedge \forall x_1\forall x_2\bigl(\mathit{diff}(x_1,x_2)\rightarrow \chi_2(x_1,x_2)\bigr)
\wedge \Phi,$
\end{center}
we replace $\varphi''$ by
\begin{center}
$\varphi := \forall x_1\forall x_2\bigl(\mathit{diff}(x_1,x_2)\rightarrow ( \chi_1(x_1) \wedge \chi_2(x_1,x_2))\bigr)
\wedge \Phi.$
\end{center}
(We ignore the case with a one-element domain as
we can simply store and return the answer in that case.)
For the remainder of Section~\ref{fullyuniformsection}, we fix the
obtained sentence $\varphi$ and weight functions $w$ and
$\bar w$ that assign weights to each symbol $R$ in the
vocabulary $\eta$ of $\varphi$; our aim is to compute
$\mathrm{WFOMC}(\varphi,n,w,\bar{w})$. We let
\begin{equation}\label{varphi100equation}
\varphi\ =\ 
\forall x_1\forall x_2\, \psi_2\ \wedge\ 
\dots \wedge\ \forall x_1\dots \forall x_{p}\, \psi_p,
\end{equation}
so the individual matrices are 
denoted by $\psi_i$. We
denote each conjunct $$\forall x_1\dots \forall x_k\, \psi_k$$ by $\varphi_k$.
%We let $\eta$ be the vocabulary of~$\varphi$.
%
The next two lemmas are crucial for computing $\mathrm{WFOMC}(\varphi,n,w,\bar{w})$
in polynomial time.
\begin{lemma}\label{firstclaim}
$\mathfrak{M}\models\varphi$ iff for all $k\in\{2,\dots , p\}$,
we have $\mathfrak{M}_k\models\varphi_k$ for every $k$-element
submodel $\mathfrak{M}_k$ of\, $\mathfrak{M}$. 
\end{lemma}
\begin{proof}
The first implication is immediate since universal
sentences are preserved under taking submodels.
For the converse implication, assume that for all $k\in\{2,\dots , p\}$, $\mathfrak{M}_k\models\varphi_k$ for all
submodels $\mathfrak{M}_k$ of $\mathfrak{M}$ of size $k$. Assume that $\mathfrak{M}\not\models\varphi$.
Thus $\mathfrak{M}\not\models\varphi_k$ for some $k$. The matrix $\psi_k$ of $\varphi_k$ is of
the type $\mathit{diff}(x_1,\dots x_k)\rightarrow \psi$, so there exists some $k$-element 
submodel $\mathfrak{M}_k$ of $\mathfrak{M}$
with domain $\{u_1,\dots , u_k\}$ such that $\mathfrak{M}_k\not\models\psi_k(u_1,\dots , u_k)$.
This is a contradiction, so $\mathfrak{M}\models\varphi$.
\end{proof}
Let $\mathfrak{M}$ and $\mathfrak{M}'$ be $\eta$-models such that $\mathfrak{M'}$ is
obtained by changing exactly one fact of span size $k$ from positive to negative or
vice versa. Let $S$ be the $k$-element set spanned by that fact.
Then $\mathfrak{M}$ and $\mathfrak{M}'$ are \emph{$S$-variants} of each other.
%(and $\mathfrak{M}$ an $S$-variant of $\mathfrak{M}'$).
%

%
\begin{lemma}\label{secondclaim}
Let $\mathfrak{M}$ and $\mathfrak{M}'$ be $S$-variants of
each other, $|S| > 1$. Let $U\not= S$ be a set of elements of\, $\mathfrak{M}$
such that $|U| = m > 1$. Let $\mathfrak{M}_U$ and $\mathfrak{M}_U'$ be the
submodels of\, $\mathfrak{M}$ and $\mathfrak{M}'$
induced by $U$. Then $\mathfrak{M}_U\models\varphi_m$ iff\, $\mathfrak{M}_U'\models\varphi_m$.
\end{lemma}
\begin{proof}
Firstly, if the formula $\varphi_m = \forall x_1\dots \forall x_m\, \psi_m$
contains atoms of arity two or more, then, by the
syntactic restrictions of $\mathrm{SU}_1$, each of those atoms mentions exactly all of
the variables $x_1,\dots x_m$. Secondly, $\psi_m$ is of the form
$\mathit{diff}(x_1,\dots , x_m)\rightarrow \psi$.
\end{proof}
%

%We are now ready to prove the main result of the section.

%
%
\begin{lemma}\label{maintech}
The weighted model counting problem for
each $\mathrm{SU}_1$-sentence is in $\mathrm{PTIME}$.
\end{lemma}
\begin{proof}
As discussed above, we prove the claim for the sentence $\varphi$ we
have fixed. Let $T$ be the set of $1$-types over the vocabulary $\eta$ of $\varphi$.
Fix an ordering of $T$ and let $\alpha_1,\dots , \alpha_{\ell}$ enumerate $T$ in that order. 
For a positive integer $k = \{0,\dots , k-1\}$, a function $f:k \rightarrow T$ is a 
\emph{type assignment} over $k$.
Two type assignments $f:k\rightarrow T$ and $g:k\rightarrow T$ are 
said to have the \emph{same multiplicity}, if for each $\alpha\in T$,
the functions $f$ and $g$ map the same number of elements in $k$ to $\alpha$.
%

%
%
%
%and with respect to $\eta$.
%
%
%
%Consider the conjunct $\varphi_k = \forall x_1\dots \forall x_k\, \psi$ of $\varphi$ with $k$ variables.
%and let $\eta'\subseteq\eta$ be the vocabulary of $\varphi_k$.
For a type assignment $f:k\rightarrow T$,
let $\mathcal{M}_{f,k}$ be the set of all $\eta$-models $\mathfrak{M}$
such that the following conditions hold.
\begin{enumerate}
\item
The domain of $\mathfrak{M}$ is $k = \{0,\dots , k-1\}$, and the
size of the span of each
positive fact $Ru_1\dots u_{m}$ of $\mathfrak{M}$ is either $1$ or $k$, i.e.,
each positive fact either spans a single domain
element or all of the domain elements of $\mathfrak{M}$.
\item
For each $m\in\{0,\dots , k-1\}$, we have $\mathfrak{M}\models \alpha_{f(m)}(m)$.
\item
$\mathfrak{M}\models\varphi_k$.
\end{enumerate}
Recalling the relativised weight function $\mathrm{W}_k$
from the preliminaries, we define
the \emph{local weight} $\mathit{lw}(\varphi_k,f)$ of $\varphi_k$ 
with respect to a type assignment $f:k\rightarrow T$ so that 
$$\mathit{lw}(\varphi_k,f)\, :=\, \sum\limits_{\mathfrak{M}\,
\in\, \mathcal{M}_{f,k}}W_{k}(\mathfrak{M},w,\bar{w}).$$
Thus $\mathit{lw}(\varphi_k,f)$
could be characterized as giving the weighted number of models of $\varphi_k$
with domain $k$ and with $1$-types distributed 
according to $f$ so that 
only those positive and negative facts are counted that have span $k$.
Clearly $\mathit{lw}(\varphi,f) = \mathit{lw}(\varphi,g)$
for any $g:k\rightarrow T$ that has the same multiplicity as $f$, so only the
number of realizations of the $1$-types matters rather than the
concrete
realizations.
Therefore we define, for any 
nonnegative integers $k_1,\dots , k_{\ell}$
such that $k_1+\dots + k_{\ell} = k$, that  
$\mathit{lw}(\varphi_k,(k_1,\dots , k_{\ell}))\, :=\, \mathit{lw}(\varphi_k,h),$
\begin{comment}
\smallskip

$\text{ }\ \ \ \ \ \ \ \ \ \ \ \ \ \ \ \ 
\ \ \ \ \mathit{lw}(\varphi_k,(k_1,\dots , k_{\ell}))\, :=\, \mathit{lw}(\varphi_k,h),$

\smallskip
\end{comment}
%
%
%
\noindent
where $h:T\rightarrow k$ is a 
type assignment that maps, for each $i \in [\ell]$,
precisely $k_i$ elements of $k$ to $\alpha_i$. 
Note that there exist only finitely many numbers $\mathit{lw}(\varphi_k,(k_1,\dots , k_{\ell}))$
such that $k\in\{2,\dots ,p\}$ and $k_1+\dots + k_{\ell} = k$.
We can thus compile a look-up table of these finitely many local weights.
For each tuple $(n_1,\dots , n_{\ell})$ of nonnegative integers
such that $n_1+\dots + n_{\ell} = n$, fix a unique type assignment $h:n\rightarrow T$
that maps exactly $n_i$ elements of $n$ to $\alpha_i$ for each $i\in [\ell]$.
Then, using $h$, define $\mathcal{M}_{(n_1,\dots , n_{\ell})}$
to be the class of $\eta$-models with domain $n$
where exactly the elements $i$ such that $h(i) = \alpha_i$, realize $\alpha_i$.
Clearly $\mathrm{WFOMC}(\varphi,n,w,\bar{w})$ is now given by
%
%
%
%\begin{multline}\label{12001111}kkkk
%\mathrm{WFOMC}(\varphi,n,w,\bar{w})\\
%=
\begin{equation}\label{12001111}
\sum\limits_{n_1+\dots + n_{\ell}\, =\, n}\binom{n}{n_1,\dots , n_{\ell}}
\scalebox{0.97}[1]{$\mathrm{WFOMC}(\varphi,n,w,\bar{w})
\upharpoonright {\mathcal{M}_{(n_1,\dots , n_{\ell})}}.
$}
\end{equation}
%\end{multline}
%
%
%
Therefore, to conclude the proof, we need to find a suitable
formula for $$\mathrm{WFOMC}(\varphi,n,w,\bar{w})
\upharpoonright {\mathcal{M}_{(n_1,\dots , n_{\ell})}}.$$
We shall do that next.
For each $\alpha_i\in T$, let $w_{\alpha_i}$ be the weight of 
the type $\alpha_i$.
Let $k_1,\dots , k_{\ell}$ be nonnegative integers 
that sum to $k\leq n$. A $k$-element set with $k_i$ realizations of $\alpha_i$ for
each $i\in [\ell]$ can be chosen in $\binom{n_1}{k_1}
\boldsymbol{\cdot}\, \ldots\,
\boldsymbol{\cdot}\binom{n_{\ell}}{k_{\ell}}$ ways from the 
set $n$ with $n_i$ realizations of $\alpha_i$ fixed for each $i\in [\ell]$.
By Lemmas \ref{firstclaim} and \ref{secondclaim}, we thus see that
%$\mathrm{WFOMC}(\varphi,n,w,\bar{w})
%\upharpoonright {\mathcal{M}_{(n_1,\dots , n_{\ell})}}$ is given by
%
%
%
%\[
%\begin{array}{ll}
%
%
%
\begin{multline}
\mathrm{WFOMC}(\varphi,n,w,\bar{w})
\upharpoonright {\mathcal{M}_{(n_1,\dots , n_{\ell})}}
=\ 
\bigl(\prod\limits_{i\, \leq\, \ell}(w_{\alpha_i})^{n_i}\bigr)\text{ }\hspace{6.8cm}\text{  }\\
\boldsymbol{\cdot}\prod\limits_{2\, \leq\, k\, \leq\, p}\
\prod\limits_{k_1+\dots + k_{\ell}\, =\, k}
\scalebox{0.97}[1]{$\mathit{lw}(\varphi_k , (k_1,\dots , k_{\ell}))$}^
{\binom{n_1}{k_1}\boldsymbol{\cdot}\,
\dots\ \boldsymbol{\cdot}\binom{n_{\ell}}{k_{\ell}}}.
\end{multline}
%
%
%
\begin{comment}
\begin{multline}
\mathrm{WFOMC}(\varphi,n,w,\bar{w})
\upharpoonright {\mathcal{M}_{(n_1,\dots , n_{\ell})}}
=\ 
\bigl(\prod\limits_{i\, \leq\, \ell}(w_{\alpha_i})^{n_i}\bigr)\text{ }\hspace{6.8cm}\text{  }\\
\boldsymbol{\cdot}\prod\limits_{2\, \leq\, k\, \leq\, p}\
\sum\limits_{k_1+\dots + k_{\ell}\, =\, k}
\binom{n_1}{k_1}\boldsymbol{\cdot}\, \dots\ \boldsymbol{\cdot}\binom{n_{\ell}}{k_{\ell}}
\scalebox{0.97}[1]{$\mathit{lw}(\varphi_k , (k_1,\dots , k_{\ell})).$}
\end{multline}
\end{comment}
%
%
%
%\end{array}
%\]
%

%
\begin{comment}
A \emph{$1$-typed domain} is an $\eta$-model where all
facts of span $k\not= 1$ are fixed negatively, so only facts of span $1$
are freely fixed. A $k$-expansion of a $1$-typed domain is an $\eta$-model obtained
from a $1$-typed domain by freely modifying the facts of span $k$, i.e., possibly 
making some of the negative facts of span $k$ positive.
\end{comment}
%

%
%
%
%
%
%
%

%
%
%
\begin{comment}
We define as follows a function $N_k$ that gives the
total weight of $k$-expansions $\mathfrak{M}$ of a $1$-typed domain
with $n_i$ realizations of $\alpha_i$ for each $i$ so that $\mathfrak{M}\models\varphi_k$.
\end{comment}
%

%
Therefore the function in Line (\ref{12001111}) 
can clearly be computed in $\mathrm{PTIME}$ in $n$
(which is given in unary).
%when $\varphi$, $w$ and $\bar{w}$ are constant.
%
%
%
\end{proof}
\subsection{Weighted model counting for $\mathrm{U}_1$}\label{thiswasthe}
%

%
%Recall from the preliminaries section that $\mathrm{U}_1$ can 
%express properties not expressible in $\mathrm{FU}_1$. 
%This extra expressivity of $\mathrm{U}_1$ obviously  
%comes from the more relaxed syntax.
%Due to Lemmas \ref{scottlemma2} and \ref{skolemlemma}, the
As discussed in the preliminaries, the weighted model
counting problem of $\mathrm{U}_1$-sentences can be reduced to
the corresponding problem for
conjunctions of $\forall^*$-sentences of $\mathrm{U}_1$.
A natural next step would be to
follow the strategy of Section \ref{fullyuniformsection}.
However, that approach would fail due to Lemma \ref{secondclaim}
which depends crucially on the exact syntactic properties of $\mathrm{SU}_1$.
%
%
%
\begin{comment}
the proof of Theorem \ref{maintech} would fail---as the reader
can easily check---because the argument depends crucially on
the individual $\forall^*$-conjuncts being $\mathrm{SU}_1$-formulae, so
that any higher arity atom \emph{must} contain \emph{all} the variables in
the quantifier prefix whose scope that atom is in.
\end{comment}
%
%
%
%
%
%
%
Thus we need a different approach.
We now show how to reduce the
weighted model counting problem for $\mathrm{U}_1$ 
to the corresponding problem for $\mathrm{SU}_1$.
%considered in the previous section.
%
%

%
\begin{comment}
%
%Let $\forall x_1 \dots \forall x_k\, \varphi$ be an $\forall^*$-sentence of $\mathrm{U}_1$.
Consider a $\mathrm{U}_1$-matrix $\varphi$. (Recall that a $\mathrm{U}_1$-matrix is
quantifier-free by definition.)
Let $\eta$ be the vocabulary of $\varphi$.
The matrix $\varphi$ is in \emph{type form} if
%
%
%
$$\varphi\ = \ \alpha_1(x_1)\wedge\dots \wedge\alpha_k(x_{\ell})
\ \wedge\ \beta(y_1,\dots , y_k)\ \wedge\ \mathit{diff}(x_1,\dots , x_k),$$
%
%
%
where each $\alpha_i$ is a $1$-type over $\eta$
and $\beta(x_1,\dots , x_k)$ a $k$-table over $\eta$ in
the variables $\{y_1,\dots , y_k\}\subseteq \{x_1,\dots , x_{\ell}\}$.
The matrix $\varphi$ is in type form also in the special case where
the $k$-table above is replaced by $\top$.
%
%
%$\varphi\ = \ \alpha_1(x_1)\wedge\dots \wedge\alpha_k(x_{\ell})
%\wedge \mathit{diff}(x_1,\dots , x_k),$ where the formulae are as
%above but there is no $k$-table at all.
%
An $\exists^*$-sentence $\exists x_1\dots \exists x_{\ell}\, \varphi$ 
%
of $\mathrm{U_1}$ is in \emph{$\exists^*$-type form} if $\varphi$ is in
%
type form with respect to $x_1,\dots , x_{\ell}$.
%

%
%
%
\begin{lemma}\label{type-formlemma}
%
Every $\exists^*$-sentence of $\mathrm{U}_1$\, with a 
disjunction-free matrix is equivalent to a 
disjunction of $\exists^*$-type-form formulae.
%
\end{lemma}
%
\end{comment}
%

We begin with the Lemma \ref{utosu} below.
Restricting attention to $\forall^*$-sentences in the lemma is
crucial, since $\mathrm{SU}_1$ is in general
strictly less expressive than $\mathrm{U}_1$, as shown in \cite{kieku15}.

%
\begin{comment}
%
We begin with the Lemma \ref{utosu} below.
Restricting attention to $\forall^*$-sentences in the lemma is crucial, as
Theorem 3 of \cite{kieku15} shows that $\mathrm{SU}_1$ is 
strictly less expressive than $\mathrm{U}_1$. That result is established by
showing that the so-called \emph{covering-node property} separates 
$\mathrm{U}_1$ from $\mathrm{SU}_1$. The property states that a
ternary relation $R$ has some node
that belongs to every tuple of $R$, as 
expressed by the $\mathrm{U}_1$-sentence $\exists u \forall x\forall y
\forall z ( Rxyz\rightarrow (u=x\vee u=y\vee u=z))$.
An Ehrenfeucht-Fra\"{i}ss\'{e} game argument
specific to $\mathrm{SU}_1$ shows the property is not expressible in $\mathrm{SU}_1$.
Despite that, the following holds.
%
\end{comment}

%
\begin{lemma}\label{utosu}
Every $\forall^*$-sentence of\, $\mathrm{U}_1$ translates to an
equivalent Boolean combination $\forall^*$-sentences of\, $\mathrm{SU}_1$. 
\end{lemma}
\begin{proof}
We sketch the proof. See Appendix \ref{proofofSUvsUlemma} for further details.

It is easy to show that every $\exists^*$-sentence of $\mathrm{U}_1$ is
equivalent to a disjunction of $\exists^*$-sentences of the form

\smallskip

\noindent
\scalebox{0.90}[1.05]{$\exists x_1\dots \exists x_{\ell}\bigl(\alpha_1(x_1)\wedge \dots \wedge \alpha_{\ell}(x_{\ell})
\wedge \beta(x_1,\dots , x_k) \wedge \mathit{diff}(x_1,\dots , x_{\ell})\bigr),$}

\smallskip

\noindent
where $\alpha_i$ are $1$-types and $\beta$ is a $k$-table.
For this to be an $\mathrm{SU}_1$-sentence, $k$ would need to be equal to $\ell$.
However, this sentence can be seen equivalent to
the following conjunction of $\mathrm{SU}_1$-sentences:

\smallskip

\noindent
\scalebox{0.86}[1.05]{$\exists x_1\dots \exists x_k \bigl(\,
\alpha_1(x_1)\wedge\dots \wedge \alpha_k(x_k)
\wedge \beta(x_1,\dots , x_k) \wedge \mathit{diff}(x_1,\dots , x_k)\, \bigr)$}\\
$\text{ }$\ \ \ \ \ \ \ \ 
\ \ \ \ \ \ \ \ \scalebox{0.86}[1.1]{$\wedge\ \exists x_1\dots \exists x_{\ell}\bigl(\, 
\alpha_1(x_1)\wedge\dots \wedge \alpha_{\ell}(x_{\ell})
\wedge \mathit{diff}(x_1,\dots , x_{\ell})\bigr).$}
\end{proof}

%
%We are now ready to prove the following.
%
%
%
%a conjunction of \emph{universal} $\mathrm{U}_1$-sentences
%can be modified to an equivalent $\mathrm{SU}_1$-sentence which,
%however, is a Boolean combination of universal sentences. We would instead
%like to obtain a plain conjunction of universal $\mathrm{SU}_1$-sentences.
%This turns out to be straightforward to achieve.
%
%
%
\begin{theorem}\label{reductionuone}
The weighted model counting problem is in $\mathrm{PTIME}$ for
each sentence of\, $\mathrm{U}_1$.
\end{theorem}
\begin{proof}
%
%
%
%We prove the claim by showing that every $\mathrm{U}_1$-sentence 
%translates to a counting equivalent 
%conjunction of\, $\forall^*$-sentences of\, $\mathrm{SU}_1$.
%

%
%Let $\varphi$ be a $\mathrm{U}_1$-sentence.
As discussed in the preliminaries, it suffices to prove the theorem for a
conjunction $\chi$ of\, $\forall^*$-sentences of\, $\mathrm{U}_1$.
We apply Lemma \ref{utosu} to $\chi$, obtaining a
sentence $\psi\equiv \chi$
which is a Boolean combination of\, $\forall^*$-sentences of $\mathrm{SU}_1$.
%
%
%
\begin{comment}
\begin{equation}\label{uproof3}
%
\mathrm{WFOMC}(\chi,n,w,\bar{w})
%
= \mathrm{WFOMC}(\psi,n,w,\bar{w}).
%
\end{equation}
\end{comment}
%
%
%
%
%
%
%
By Lemmas \ref{scottlemma2} and \ref{skolemlemma}, we have
%
%
%
%\begin{equation}\label{uproof4}
%
$\mathrm{WFOMC}(\psi,n,w,\bar{w})
= \mathrm{WFOMC}(\mathit{Sk}(\mathit{Sc}(\psi)),n,w',\bar{w}'),$
%
%\end{equation}
%
%
%
where $w'$ and $\bar{w}'$ are obtained from $w$ and $\bar{w}$ by
mapping the new symbols as specified in the lemmas.
%introduced by the
%operator $\mathit{Sc}$ (respectively, $\mathit{Sk}$) to $1$ (resp., $-1$).
$\mathit{Sk}(\mathit{Sc}(\psi))$ is an $\forall^*$-sentence of $\mathrm{SU}_1$.
\end{proof}
%

%
\begin{comment}
From this Lemma and the considerations in 
the previous sections, we obtain the following theorem directly.

%
\begin{theorem}
%
The weighted model counting problem is in $\mathrm{PTIME}$ for
each sentence of $U_1$.
%
\end{theorem}
\end{comment}
%

\section{Counting and prefix classes}

First-order prefix classes admit the following neat classification:

\begin{proposition}\label{prefixclassproposition}
Consider a prefix class $C_w$ of first-order logic defined by a 
quantifier-prefix $w\in \{\exists,\forall\}^*$.
\begin{enumerate}
\item
If $|w|\geq 3$, then $C_w$ contains a formula with a $\#\mathrm{P}_1$-complete
symmetric weighted model counting problem.
\item
If $|w| < 3$, then the symmetric weighted model counting
problem of each formula in $C_w$ is in $\mathrm{PTIME}$.
\end{enumerate}
\end{proposition}

We note that the proof of the Proposition makes use of the results and 
techniques of \cite{DBLP:conf/kr/BroeckMD14, suc15} in various ways, and thus much of the 
credit goes there.
We sketch the proof---see Appendix \ref{prefixclasses} for more details.

Firstly, \cite{suc15} shows that there is an $\mathrm{FO}^3$-sentence $\varphi$ 
with a $\#\mathrm{P}_1$-complete model counting problem.
We turn $\varphi$ into a conjunction of prenex form sentences by eliminating quantified subformulae in a way resembling the Scott normal form procedure. We 
then apply the Skolemization operator $\mathit{Sk}$ (see Section \ref{sect:skolemn}).
Combining the obtained $\forall^*$-conjuncts,
we get a sentence $\chi := \forall x\forall y \forall z\psi$
with the same model counting problem as $\varphi$;
here $\psi$ is quantifier-free.

\begin{comment}
Firstly, \cite{suc15} shows that there is an $\mathrm{FO}^3$-sentence $\varphi$ 
for which the counting problem is $\#\mathrm{P}_1$-complete.
We turn $\varphi$ into a conjunction of prenex normal form sentences by eliminating quantified subformulae in a simple way that resembles the Scott normal form procedure. We 
then we apply the Skolemization operator $\mathit{Sk}$ (see Section \ref{sect:skolemn}).
Combining the obtained $\forall^*$-conjuncts
together, we get a sentence $\chi := \forall x\forall y \forall z\psi$
with the same ($\#\mathrm{P}_1$-complete) counting problem as $\varphi$;
here $\psi$ is quantifier-free.
\end{comment}
%

%
We then start modifying the $\forall\forall\forall$-sentence $\chi$ in order to obtain, for
each prefix class $C$ with three quantifiers, a
sentence in $C$ with the same model counting problem as $\chi$.
The required modifications can be easily done by using operations
that slightly generalize the Skolemization operation from 
Section \ref{sect:skolemn}. These operations are defined as follows.
Let $\chi' := \forall x_1\dots \forall x_k Q_1x_{k+1}\dots Q_m x_m\, \chi''$ be a 
prenex form sentence with $\chi''$ quantifier-free and with $Q_i\in\{\exists,\forall\}$.
We turn $\chi'$ into 
$\forall x_1\dots \forall x_k Q_1'x_{k+1}\dots Q_m' x_m(Ax_1\dots x_k\vee \neg\chi'')$,
where $A$ is a fresh $k$-ary predicate and each $Q_i'$ is the dual of $Q_i$.
The difference with the Skolemization operation of
Section \ref{sect:skolemn} is simply that $Q_1$ is not required to be $\exists$.
This new sentence has the same
model counting problem as $\chi'$ when the fresh symbol $A$ is 
given weights exactly as in Lemma \ref{skolemlemma}. The proof of this claim is 
similar to the proof of Lemma \ref{skolemlemma}.
The second claim of Proposition \ref{prefixclassproposition}
holds by the
result for $\mathrm{FO}^2$.
\begin{comment}
The second claim of Proposition \ref{prefixclassproposition}
follows immediately from the
result for $\mathrm{FO}^2$ from
\cite{DBLP:conf/kr/BroeckMD14}.
\end{comment}

\section{Conclusions}

It can be shown that $\mathrm{WFOMC}$ for formulae of two-variable logic
with counting $\mathrm{C}^2$ can be reduced to $\mathrm{WFOMC}$ for $\mathrm{FO}^2$
with \emph{several} functionality axioms. Proving tractability in that setting remains an
interesting open problem. One difficulty here is that the interaction patterns of 
different functional relations cause effects that could intuitively be
described as `non-local' and seem to require significantly more general
combinatorial arguments than those in Section \ref{funcounting}.
The tools of \cite{kopczynski} could prove useful here.

\medskip

\medskip

\noindent
\textbf{Acknowledgments.} The authors were supported by the
  ERC-grant~{647289} `CODA.'

%% Bibliography
%\bibliography{bibfile}

\bibliography{count} 

\begin{thebibliography}{10}

\bibitem{DBLP:books/daglib/0041477}
Franz Baader, Ian Horrocks, Carsten Lutz, and Ulrike Sattler.
\newblock {\em An Introduction to Description Logic}.
\newblock Cambridge University Press, 2017.

\bibitem{guardednegation}
Vince B{\'{a}}r{\'{a}}ny, Balder ten Cate, and Luc Segoufin.
\newblock Guarded negation.
\newblock {\em J. {ACM}}, 62(3):22:1--22:26, 2015.

\bibitem{suc15}
Paul Beame, Guy~Van den Broeck, Eric Gribkoff, and Dan Suciu.
\newblock Symmetric weighted first-order model counting.
\newblock In {\em 34th {ACM} Symposium on Principles of Database Systems
  (PODS)}, pages 313--328, 2015.

\bibitem{DBLP:conf/kr/BroeckMD14}
Guy~Van den Broeck, Wannes Meert, and Adnan Darwiche.
\newblock Skolemization for weighted first-order model counting.
\newblock In {\em Principles of Knowledge Representation and Reasoning ({KR})},
  2014.

\bibitem{suciusurvey}
Guy~Van den Broeck and Dan Suciu.
\newblock Query processing on probabilistic data: {A} survey.
\newblock {\em Foundations and Trends in Databases}, 7(3-4):197--341, 2017.

\bibitem{DBLP:series/synthesis/2009Domingos}
Pedro~M. Domingos and Daniel Lowd.
\newblock {\em Markov Logic: An Interface Layer for Artificial Intelligence}.
\newblock Synthesis Lectures on Artificial Intelligence and Machine Learning.
  Morgan {\&} Claypool Publishers, 2009.

\bibitem{ebbinghausflum}
Heinz{-}Dieter Ebbinghaus and J{\"{o}}rg Flum.
\newblock {\em Finite model theory}.
\newblock Perspectives in Mathematical Logic. Springer, 1995.

\bibitem{kuusihella}
Lauri Hella and Antti Kuusisto.
\newblock One-dimensional fragment of first-order logic.
\newblock In {\em Advances in Modal Logic 10}, pages 274--293, 2014.

\bibitem{kazemi}
Seyed~Mehran Kazemi, Angelika Kimmig, Guy~Van den Broeck, and David Poole.
\newblock New liftable classes for first-order probabilistic inference.
\newblock In {\em Annual Conference on Neural Information Processing Systems
  ({NIPS})}, pages 3117--3125, 2016.

\bibitem{kuuki}
Emanuel Kieronski and Antti Kuusisto.
\newblock Complexity and expressivity of uniform one-dimensional fragment with
  equality.
\newblock In {\em Mathematical Foundations of Computer Science ({MFCS}) Part
  {I}}, pages 365--376, 2014.

\bibitem{kieku15}
Emanuel Kieronski and Antti Kuusisto.
\newblock Uniform one-dimensional fragments with one equivalence relation.
\newblock In {\em Annual Conference on Computer Science Logic ({CSL})}, pages
  597--615, 2015.

\bibitem{kimmig}
Angelika Kimmig, Guy~Van den Broeck, and Luc~De Raedt.
\newblock Algebraic model counting.
\newblock {\em J. Applied Logic}, 22:46--62, 2017.

\bibitem{kopczynski}
Eryk Kopczynski and Tony Tan.
\newblock Regular graphs and the spectra of two-variable logic with counting.
\newblock {\em {SIAM} J. Comput.}, 44(3):786--818, 2015.

\bibitem{kuusisurvey}
Antti Kuusisto.
\newblock On the uniform one-dimensional fragment.
\newblock In {\em International Workshop on Description Logics ({DL})}, 2016.

\end{thebibliography}
%\bibliographystyle{plain}

%% Appendix
\appendix

\section{Appendix}

\subsection{Scott normal forms}\label{scottnormalforms}

Here we briefly discuss the principal properties of the reduction of
formulae to Scott normal form. The process is well-known, so we
only sketch the related details.
Let $\varphi$ be a sentence of $\mathrm{U}_1$. 
Note that $\mathrm{FO}^2$ and of course $\mathrm{SU}_1$
are syntactic fragments of $\mathrm{U}_1$. To put $\varphi$ into
generalized Scott normal form, consider a
subformula $\psi(x) = Qy_1\dots Q{y_k}\chi(x,y_1,\dots , y_k)$ of $\varphi$,
where $Q\in\{\forall, \exists\}$
and $\chi$ is quantifier-free. Now, $\psi(x)$ has one
free variable. Thus we let $P_{\psi}$ be a fresh unary 
predicate and consider the sentence 
$$\forall x( P_{\psi}x\ \leftrightarrow\ Qy_1\dots Q{y_k}\chi(x,y_1,\dots , y_k))$$
which states that $\psi(x)$ is equivalent to $P_{\psi}x$.
Letting $Q'$ denote the dual of $Q$, i.e., $Q' = \{\exists,\forall\}\setminus \{Q\}$,
this sentence is seen equivalent to 
\begin{multline}
\chi'\ :=\ \forall xQy_1\dots Q{y_k}( P_{\psi}x\ \rightarrow \chi(x,y_1,\dots , y_k))\\
\wedge \forall xQ'y_1\dots Q'{y_k}( \chi(x,y_1,\dots , y_k)\rightarrow P_{\psi}x).
\end{multline}
Therefore $\varphi$ is has the same weighted model count as the sentence 
%
%
%
%\begin{multline}
%
$$\chi''\ =\ \chi'\ \wedge\ \varphi[P_{\psi}(x)/\psi(x)],$$
%
%\end{multline}
%
%
%
where $\varphi[P_{\psi}(x)/\psi(x)]$ is
obtained from $\varphi$ by replacing $\psi(x)$ with $P_{\psi}(x)$; the fresh 
relation symbol $P$ is given the weight $1$ in both positive and negative facts.
Repeating this, we eliminate quantifiers one by one, starting from the atomic
level and working upwards from there. We always introduce a new predicate
symbol ($P_{\psi}$ in the above example) and axiomatize that symbol to be
equivalent to the formula beginning with the quantifier to be 
eliminated ($\psi(x)$ in the above example).

Note that while $\psi(x)$ had a 
free variable, we may also need to eliminate quanfiers from
subformulae without free variables, such as, e.g., $\exists x Ax$. 
Then a fresh nullary predicate needs to be introduced.
%We always eliminate a maximally long block of
%quantifiers $Q_{y_1}\dots Q y_{k}$ without alternations,
Note that quantifying in $\mathrm{U}_1$ leaves at most one
free variable, so the fresh symbols are always at most unary by
%The fresh relation symbol must be either unary or nullary by
the definition of the syntax of $\mathrm{U}_1$. We clearly end up
with a sentence in generalized Scott normal form.

We make the following observations
\begin{enumerate}
\item
The Scott-normal form version $\mathit{Sc}(\varphi)$ of a
sentence $\varphi$ indeed has the required property
that $\exists P_1\dots \exists P_m\, \mathit{Sc}(\varphi)$ is
equivalent to $\varphi$, where $P_1,\dots , P_{m}$ are the
fresh unary and nullary predicates.
\item
If $\varphi$ is a sentence of $\mathrm{U}_1$ (respectively, $\mathrm{SU}_1$, $\mathrm{FO}^2$),
then the sentence $\mathit{Sc}(\varphi)$ is a sentence of $\mathrm{U}_1$
(respectively, $\mathrm{SU}_1$, $\mathrm{FO}^2$). This is easy to 
see by first noting that the fresh symbols are unary or nullary, and noting then that the 
syntax of $\mathrm{U}_1$ allows free use of unary and nullary symbols.
\item
We have $\mathrm{WFOMC}(\varphi,n,w,\bar{w})
\ \ =\ \ \mathrm{WFOMC}(\mathit{Sc}(\varphi),n,w',\bar{w}')$,
where $w$ and $\bar{w}$ map the fresh symbols to $1$.
The reason for this is that the novel symbols are axiomatized to be
equivalent to the unary and nullary formulae, and thereby the novel
symbols must have a unique interpretation in each model of $\mathit{Sc}(\varphi)$.
\item
In the case of $\mathrm{FO}^2$, the novel sentences
$\forall x \forall y \chi$ that arise when axiomatizing the 
fresh predicates can be pushed together so that only a
single $\forall^*$-conjunct $\forall x \forall y \chi'$ rather
than a conjunction $\forall x \forall y \chi_1\wedge\dots \wedge \forall x \forall y \chi_n$ 
will be part of the ultimate Scott normal form formula. 
\end{enumerate}
\subsection{Proof of Lemma \ref{skolemlemma}}\label{skolemlemmaproof}
Before proving Lemma \ref{skolemlemma}, we
define that a \emph{self-inverse bijection} is an
involutive bijection, so $f(f(x)) = x $ for all $x\in \mathit{dom}(f)$.
We then prove the lemma.
\begin{proof}
We will consider the formulae
\begin{align*}
\chi_1 :=&\ \forall x_1\dots \forall x_k
\exists y_1\dots \exists y_m Q_1z_1\dots Q_n z_n\, \psi\\
\chi_2 :=&\ \forall x_1\dots \forall x_k(Ax_1\dots x_k \vee \neg
\exists y_1\dots \exists y_m Q_1z_1\dots Q_n z_n\, \psi)
%
%
%
%\chi_2 :=&\ \forall x_1\dots \forall x_k\forall y_1\dots
%\forall y_m Q_1'z_1\dots Q_k' z_k\, (Ax_1\dots x_k \vee \psi)
%
\end{align*}
from our definition of Skolemization and 
show the following:
\begin{equation}\label{toomanycomplications}
\mathrm{WFOMC}(\chi_1,n,v,\bar{v})\upharpoonright{\{\mathfrak{B}\}}
=\mathrm{WFOMC}(\chi_2,n,v',\bar{v}')\upharpoonright{\mathcal{C}}
\end{equation}
where $v'$ and $\bar{v}'$ extend $v$ and $\bar{v}$ on the 
input $A$ such that $v'(A) = 1$ and $\bar{v}'(A) = -1$,
and $\{\mathfrak{B}\}$ is a singleton model class 
where $\mathfrak{B}$ is a $\mathit{voc}(\chi_1)$-model and $\mathcal{C}$
the model class $\{\, (\mathfrak{B}, A)\, |\, A\subseteq \mathit{dom}(\mathfrak{B}) = n\, \}$.
Assume $\mathfrak{B}\models\chi_1$.
Then an expanded model $(\mathfrak{B},A)$ satisfies $\chi_2$ if and only if $A$ is
interpreted to be the total $k$-ary relation over the domain $n$.
Thus Equation \ref{toomanycomplications} holds.
Assume then that $\mathfrak{B}\not\models\chi_1$.
We will show that the sum of the weights of the models in $\mathcal{C}$
that satisfy $\chi_2$ is zero.
This will conclude the proof.

Let $U$ be the set of tuples $(u_1,\dots, u_k)\in n^k$ such that
$$\mathfrak{B}\models \exists y_1\dots \exists y_m Q_1z_1
\dots Q_n z_n\, \psi\, (u_1,\dots , u_k).$$
We have $(n^k \setminus U)\not=\emptyset$ as $\mathfrak{B}\not\models\chi_1$.
Let $\mathcal{M}$ be the class of models in $\mathcal{C}$ that satisfy $\chi_2$.
As models $\mathfrak{N}\in\mathcal{M}$ must satisfy $\chi_2$,
each $\mathfrak{N}\in\mathcal{M}$ has $A^{\mathfrak{N}}\supseteq U$.
Furthermore, for each $A'\supseteq U$ such that $A'\subseteq n^k$, there clearly exists a 
model $\mathfrak{N}'\in\mathcal{M}$ so that $A^{\mathfrak{N}'} = A'$.
We shall define a self-inverse bijection $f:\mathcal{M}\rightarrow\mathcal{M}$
such that the weights of $\mathfrak{N}$ and $f(\mathfrak{N})$ cancel 
for each $\mathfrak{N}\in\mathcal{M}$, thereby concluding the proof.

Let $\overline{u}$ be the lexicographically smallest tuple in $(n^k \setminus U)\not= \emptyset$
(we have $(n^k \setminus U)\subseteq n^k$, so a lexicographic ordering is defined).
We define $f$ so that it sends each model $\mathfrak{N}\in\mathcal{M}$ to the model where $A$ is
modified simply by changing the interpretation of $A$ on $\overline{u}$: if $A$ is
true on $\overline{u}$, we make it false, and if $A$ is false on $\overline{u}$, we make it true, and
on other tuples, we keep $A$ the same. It is thus clear that the 
weights of any $\mathfrak{N}\in\mathcal{M}$ and $f(\mathfrak{N})$ cancel each other, as
the models differ only on the interpretation of $A$ on this one tuple (and $\bar{v}(A) = -1$).
%
%
%
%Therefore $\mathrm{WFOMC}(\chi_1,n,v,\bar{v})
%\leq\mathrm{WFOMC}(\varphi,n,v',\bar{v}')$
%
%
%
%
%
%
\end{proof}

\begin{comment}
Let $\varphi$ be a $\mathrm{U}_1$-sentence.
By Lemma \ref{scottlemma2},
we have
%
%
%
\begin{equation}\label{uproof1}
%
\mathrm{WFOMC}(\varphi,n,w,\bar{w})
%
= \mathrm{WFOMC}(\mathit{Sc}(\varphi),n,w',\bar{w}'),
%
\end{equation}
%
%
%
where $w$ and $\bar{w}'$ map the fresh unary symbols to $1$.
By Lemma \ref{skolemlemma}, we have
%
%
%
\begin{equation}\label{uproof2}
%
\mathrm{WFOMC}(\mathit{Sc}(\varphi),n,w',\bar{w}')
%
= \mathrm{WFOMC}(\mathit{Sk}(\mathit{Sc}(\varphi)),n,w'',\bar{w}''),
%
\end{equation}
%
%
%
where $w''$ and $\bar{w}''$ are obtained from $w'$ and $\bar{w}'$ by
mapping the fresh symbols introduced by the
operator $\mathit{Sk}$ to $-1$.
\end{comment}

\subsection{Normal forms for $\mathrm{FO^2}$ with a
functionality axiom}\label{appendixmillion}

Here we discuss how the 
sentence $\varphi\wedge \forall x\exists^{=1}y\, \psi(x,y)$ given in 
Section \ref{syntaxsection} can be modified in order to obtain the
desired normal form sentence.
We first consider only the subformula $\psi(x,y)$, ignoring $\varphi$ for awhile. 
We apply the Scott normal form
procedure for eliminating quantified subformulae
(see Appendix \ref{scottnormalforms})
to the open formula $\psi(x,y)$.
We thereby obtain from $\psi(x,y)$ a formula $$\psi'(x,y)
\wedge \forall x\forall y\, \psi''\wedge 
\bigwedge_i\forall x \exists y\, \psi_i$$ where $\psi'$, $\psi''$ and
each $\psi_i$ are quantifier-free.
We then observe that $$\forall x \exists^{=1} y(\psi'(x,y)
\wedge \forall x\forall y \psi''\wedge 
\bigwedge_i\forall x \exists y\psi_i)$$ is
equivalent to $$
\forall x\forall y \psi''\wedge 
\bigwedge_i\forall x \exists y\psi_i\ \wedge\ \forall x \exists^{=1} y\psi'(x,y).$$
We then use the Skolemization operator $\mathit{Sk}$ to the formulae $\forall x \exists y \psi_i$
and combine the resulting $\forall^*$-sentences with
each other and with $\forall x \forall y \psi''$, thereby obtaining a
sentence $\forall x\forall y \psi'''\wedge \forall x \exists^{=1}y\psi'(x,y)$.

We then modify (the so far ignored sentence) $\varphi$. We put it in Scott normal form first and
then use Skolemization, thereby obtaining a 
conjunction $\forall x \forall y \chi_1\wedge ... \wedge \forall x\forall y \chi_k$.
We combine these conjuncts with $\forall x\forall y \psi'''$ to form a
single $\forall\forall$-conjuct $\forall x \forall y \psi''''$. The ultimate 
sentence is thus $\forall x \forall y \psi'''' \wedge \forall x \exists^{=1}y\psi'(x,y)$,
where $\psi''''$ and $\psi'(x,y)$ are quantifier-free, as desired.

\subsection{Relation symbol arities in the two-variable context}\label{higherarityappendix}

If $\forall x \forall y \varphi_1 \wedge
\forall x \exists y^{=1} \varphi_2(x,y)$ contains no binary relation, we
replace $\varphi_1$ by $Rxy \wedge \varphi_1$ and give $R$ 
the weights $w(R) = \bar{w}(R) = 1$. Now $R$ must have a unique interpretation in
every model of $\forall x \forall y (Rxy \wedge \varphi_1) \wedge
\forall x \exists y^{=1} \varphi_2(x,y)$ and thus contributes nothing to 
the ultimate weighted model count.

\begin{comment}
Let $\chi_1 := Rxy \wedge \varphi_1$ and 
now consider the formula $\forall x \forall y \chi_1 \wedge
\forall x \exists y^{=1} \varphi_2(x,y)$.
\end{comment}
We then discuss the assumption that we can limit attention to 
formulae without relation symbols of arities $k>2$ when
studying the data complexity of weighted model counting 
for two-variable logic with a functionality axiom.
We first give a short justification of the assumption 
and then look at the issue in a bit more detail.
So, to put it short, the analysis of Section \ref{funcounting} will work as
such even if relation symbols of arities $k>2$ are included, the only 
difference being that the ultimate model count must be multiplied by a
(non-constant) factor $N$ that takes into account facts and negative facts of 
\emph{span sizes greater than $2$}. This factor $N$ is very easy to compute, as
our logic---using two variables---is fully invariant under changing facts with span sizes
greater than $2$ elements.

We then look at the matter in a bit more detail.
Let us first fix some sentence $$\chi := \forall x \forall y \chi' \wedge
\forall x \exists y^{=1} \chi''(x,y)$$
containing at least one relation symbol of arity $k > 2$.
Now, notice that $2$-tables are allowed to contain 
atoms such as $Rxyxxy$, while $1$-types are allowed to
contain atoms $Sxxx$ etcetera. Thus the reader can easily 
check that everything in Section \ref{funcounting} works as 
such if we allow relation symbols of arities $k>2$,
with only the following exception: the number 
$\mathrm{WFOMC}(\chi,n,w,\bar{w}) = q\in \mathbb{Q}$ obtained by
our analysis must be multiplied by $N$, which is a factor arising from the 
simple fact that 
%
%
%
%\begin{equation}
%
$\mathfrak{M}\models \chi\ \ 
\Leftrightarrow 
\mathfrak{N}\models \chi$
%
%\end{equation}
%
%
%
for all $\mathfrak{M}$ and $\mathfrak{N}$ which 
\emph{differ only in facts and negative facts
with span sizes greater than $2$}. Our analysis takes into 
account only facts of span sizes up to $2$.
We consider an example to illustrate the issue.
Assume $\chi$ contains a $k$-ary symbol $R$, with $k>2$, and all
other symbols in $\chi$ are at most binary. We show how to compute the
factor $N$.
Let $n\geq k$ be a model size. There are $n^k$ tuples of length $k$ with elements from $n$.
Exactly $n$ of these tuples have span $1$ (e.g., a tuple of 
type $(u, \dots , u)$ with $u$ repeated $k$ times).
Exactly ${\binom{n}{2}}\cdot (2^k - 2)$ of the $n^k$ tuples
have span $2$ (e.g., a $k$-tuple of type $(u,v,u,v,\dots ,u,v)$ if $k$ is even). Thus there are 
$$p(n,k)\  :=\  n^k - n - 2^k\binom{n}{2} + 2\binom{n}{2}$$
tuples with span size greater than $2$ over the domain $n$. On some of these
tuples we can define $R$ positively and negatively on others. Thus, 
letting $w(R)$ be the positive and $\bar{w}(R)$ the
negative weight for $R$, we define 
$$N(n,k)\ :=\ \sum\limits_{i\, \leq\, p(n,k)}
\binom{\, p(n,k)\, }{i}\ (w(R))^i\cdot (\bar{w}(R))^{(p(n,k) - i)}.$$
While this looks nasty, we can easily evaluate it in polynomial time in
the unary input $n$. The function $N(n,k)$ provides the desired factor $N$:
we multiply the number $\mathrm{WFOMC}(\chi,n,w,\bar{w})$, which is
given by our analysis that ignores facts of span size greater than $2$, by $N(n,k)$
and thereby get the correct result. Note that we assumed $n\geq k$ simply because
models with domain size smaller than $k$ can in any case be ignored as there
are only finitely many inputs smaller than $k$ to the model counting problem, so we
can construct a look-up table for them.
It is easy to see how to expand this to cover the case where $\chi$ has
several relations of arities greater than $2$.

\section{Appendix: $\mathrm{FO}^2$ with a functionality axiom}

\subsection{Characterizing $M_{\sigma\tau}(n_{\sigma},n_{\sigma\tau},n_{\tau},
n_{\tau\sigma})$}\label{1234990appendix}

Here we give a detailed
specification of the
characterization of $$M_{\sigma\tau}(n_{\sigma},n_{\sigma\tau},n_{\tau},n_{\tau\sigma})$$ as
being the \emph{weighted number of 
ways to connect blocks} $B_{\sigma}\supseteq C_{\sigma\tau}$
and $B_{\tau}\supseteq C_{\tau\sigma}$ to each other with $2$-tables
when $|B_{\sigma}| = n_{\sigma}$, $|C_{\sigma\tau}| = n_{\sigma\tau}$,
$|B_{\tau}| = n_{\tau}$ and $|C_{\tau\sigma}| = n_{\tau\sigma}$.
We specify the weighted number $N$ of ways to
connect the blocks in the way described below. (It is worth noting here that 
we will \emph{not} compute $N$ in the way described below in the ultimate 
polynomial time algorithm.) The number $N$ is \emph{informally} the sum of all
products $W$ that can be obtained by simultaneously assigning $2$-tables to
all edges in $B_{\sigma}\times B_{\tau}$ and multiplying the
individual weights of (positive and negative) facts in
the $2$-tables such that the following conditions hold.
\begin{enumerate}
\item
In each of the simultaneous assignments, elements in the 
cells $C_{\sigma\tau}$ and $C_{\tau\sigma}$ obtain, respectively,
witnesses in $B_{\tau}$ and $B_{\sigma}$
via suitable witnessing $2$-tables. The $2$-table for the pairs in $B_{\sigma}\times B_{\tau}$
that provide witnesses for the 
elements of $C_{\sigma\tau}$ is the $2$-table of the $2$-type $\sigma$. 
Similarly, the $2$-table for the pairs in $B_\tau\times B_{\sigma}$ that provide
witnesses for the
elements of $C_{\tau\sigma}$ is the $2$-table of $\tau$.
\item
The remaining pairs in $B_{\sigma}\times B_{\tau}$ are assigned
some non-witnessing coherent $2$-table whose inverse is, likewise, not witnessing.
\end{enumerate}
%if $q = M_{\sigma\tau}(n_{\sigma},n_{\sigma\tau},n_{\tau},n_{\tau\sigma})$,
%we can specify $q$ as follows.

To define this more formally, consider the case with the below assumptions.
\begin{enumerate}
\item
$\sigma\not=\tau$.
\item
Neither $\sigma$ nor $\tau$ is both ways witnessing.
\end{enumerate}
First, we define $N=0$ if any of the following conditions is satisfied.
\begin{enumerate}
\item
$n_{\sigma\tau}\not\leq n_{\sigma}$
or $n_{\tau\sigma}\not\leq n_{\tau}$.
\item
$n_{\sigma\tau}\not= 0$ and $\sigma(2)\not=\tau(1)$.
\item
$n_{\tau\sigma}\not= 0$ and $\tau(2)\not=\sigma(1)$.
\end{enumerate}
Otherwise, let $B_{\sigma}$
and $B_{\tau}$ be disjoint sets, $|B_{\sigma}| = n_{\sigma}$
and $|B_{\tau}| = n_{\tau}$. Let $C_{\sigma\tau}\subseteq B_{\sigma}$
and $C_{\tau\sigma}\subseteq B_{\tau}$ be sets such that $|C_{\sigma\tau}| = n_{\sigma\tau}$
and $|C_{\tau\sigma}| = n_{\tau\sigma}$. Now, assume $f:C_{\sigma\tau}\rightarrow B_{\tau}$
and $g:C_{\tau\sigma}\rightarrow B_{\sigma}$ are functions that are 
nowhere inverses of each other. (We note that such functions 
need not exist, as demonstrated, for example, by the case
where $n_{\sigma\tau}\not= 0$ and $n_\tau = 0$.) There are 
precisely $n_{\sigma\tau}$ pairs in $f$ and $n_{\tau\sigma}$ pairs in $g$,
and since $f$ and $g$ are nowhere inverses, $f$ and the inverse of $g$
occupy $n_{\sigma\tau} + n_{\tau\sigma}$ edges in $B_{\sigma}\times B_{\tau}$.
Recall from the preliminaries that, if $\beta$ is a $2$-table,
then $\langle w,\bar{w}\rangle(\beta)$ 
denotes the product of the weights of the literals in $\beta$.
We let $\beta_{\sigma}$ and $\beta_{\tau}$
denote the $2$-tables of $\sigma$ and $\tau$ and define 
$$w_{f,g} := (\langle w,\bar{w}\rangle(
\beta_{\sigma}))^{n_{\sigma\tau}}\cdot (\langle w,\bar{w}\rangle(
\beta_{\tau}))^{n_{\tau\sigma}}.$$
Now, let $T$ be the set of $2$-tables $\beta$ such
that the $2$-type $\delta := \sigma(1)\beta\tau(1)$ satisfies the following conditions.
\begin{enumerate}
\item
$\delta$ is coherent.
\item
$\delta$ is not witnessing.
\item
The inverse of $\delta$ is not witnessing.
\end{enumerate}
Consider the pairs in $B_{\sigma} \times B_{\tau}$ that do not belong to $f$ or
the inverse of $g$. Let $S\subseteq B_{\sigma} \times B_{\tau}$ be the set of these pairs. 
Let $\mathcal{F}_{f,g}$ denote
the set of all functions $F:S\rightarrow T$. For each such function $F$, define 
$$w_{F} := \prod\limits_{(u,v)\, \in\, S}\langle w,\bar{w}\rangle(F((u,v)))$$
and
$$w_{f,g,F} := w_{f,g}\cdot w_F.$$
Let $\mathcal{P}$ denote the set of
triples $(f,g,F)$ where $f:C_{\sigma\tau}\rightarrow B_{\tau}$
and $g:C_{\tau\sigma}\rightarrow B_{\tau}$ are
funtions that are nowhere inverses of each other and $F\in\mathcal{F}_{f,g}$.
Define, finally, that 
$$N\, :=\, \sum\limits_{(f,g,F)\, \in\, \mathcal{P}}w_{f,g,F}.$$
The remaining cases, including the ones where $\sigma = \tau$, are 
similar in spirit and defined analogously, so we omit them here.
\subsection{Proof of Lemma \ref{nonskolemlemma}}\label{aabbccdd}

\begin{proof}
Choose some numbers $i_1,\dots , i_k$ and $j_1,\dots , j_{\ell}$ that add to $N$.
There are
$$
\binom{N}{i_1,\dots , i_k,j_1,\dots , j_{\ell}}
$$
ways to choose precisely $i_p$ edges for the
symmetric colours $p\in [k]$ and $j_q$ edges for the directed
colours $q\in [\ell]$. There are $2^{j_1+\dots + j_{\ell}}$ ways to choose an orientation for the 
directed colours. The contribution of the 
weights is then given by the
product $$\bigl(\prod\limits_{\substack{p\, \in\, [k]}}(w_p)^{i_p}\bigr)
\cdot \bigl(\prod\limits_{\substack{q\, \in\, [\ell ] }}(x_q)^{j_q}\bigr).$$
\end{proof}

\subsection{Proof of Proposition \ref{internalcounting1}}\label{aaabbccdd}

\begin{proof}
Consider the complete graph $G$ with 
the set $n$ of vertices.
%fix some non-involutive function $f:A\rightarrow n$.
By proposition \ref{nonidempcounting1}, there exist $I(m,n)$
anti-involutive functions $f:m\rightarrow n$.
Fix a single such function $f$. 
Being anti-involutive and having an $m$-element domain, the
tuples $(u,v)\in f$ cover precisely $m$ edges of $G$.
Thus the contribution of the edges covered by $f$ to the total weight of any labelling that
assigns the weight $y$ to those edges is $y^m$.
%
%
%
%\begin{equation*}
%\sum\limits_{i + j\, =\, \binom{n}{2}-m}
%\binom{\binom{n}{2}-m}{i}k^{i}\cdot (2\ell)^{j}
%\end{equation*}
%
%
% 
With $f$ fixed, the remaining $\binom{n}{2} - m$ edges (not belonging to $f$)
can by Lemma \ref{nonskolemlemma} be labelled in different ways so that they contribute the factor $L_{\empty_{k,\ell}}
\bigl(\binom{n}{2}-m,\, w_1,\dots , w_k,x_1,\dots ,x_{\ell}\bigr)$ to the total weight.
%
%
%
%So, there exist precisely $\binom{n}{2}-m$ doubletons to 
%be labelled. Consider a fixed pair $i,j$
%of non-negative integers that add to $\binom{n}{2}-m$.
%If we colour $\ell_{k}$ doubletons with $k$ symmetric 
%colours and $\ell_{k'}$ doubletons,
%
%
%
\begin{comment}
There exist $$\binom{\binom{n}{2}-m}{i}$$
%
%
%
ways to choose exactly $i$
doubletons to be labelled with \emph{undirected} colours.
With $i$ doubletons fixed, we can
label them in $k^{i}$ ways with $k$
undirected labels, and we can then label the remaining $j$ doubletons
with $\ell$ \emph{directed} colours in $(2\ell)^j$ ways.
\end{comment}
%
%
%
\end{proof}

\subsection{Proof of Proposition \ref{bipartitecounting1}}\label{JLX}

\begin{proof}
%
%
%
%Let $A$ and $B$ be disjoint sets, $|A|=M$ and $|B|=N$.
By Lemma \ref{nonidempcounting}, there exist $K(m,M,n,N)$ ways
to define a pair of functions $f,g$ so that $f:A_m\rightarrow B$ and $g:B_n\rightarrow A$
are nowhere inverses of each other. 
Now fix a pair $f,g$ of such functions. The contribution of $f$ and $g$ to the
weight of any labelling  that contains $f$ and $g$ is $y^m z^n$.
There are $MN - m - n$ edges outside the functions $f$ and $g$. (The 
pathological cases where $MN - m - n$ is negative are harmless due to the  
definition of $L_{\empty_{k, \ell}}$.)
By Lemma \ref{nonskolemlemma}, the contribution of these edges to the
total weight is $$L_{\empty_{k, \ell}}(MN - m - n,\ w_1,\dots , w_k,x_1,\dots ,x_{\ell}).$$
\end{proof}

\subsection{Proof of Proposition \ref{internalcounting2}}\label{gghhlljj}

\begin{proof}
When $m$ is even, then $m/2 = \lfloor m/2 \rfloor$ gives the 
number of those edges over the vertex set $m$ that will be part of the
complete matching of $m$. Thus
there are then $\binom{n}{2} - m/2$ edges outside the matching in the graph
with vertex set $n$. Note that $F(m)=0$ when $m$ is odd;
we write $\lfloor m/2 \rfloor$ simply to ensure the 
inputs to $L_{\empty_{k,\ell}}$ are integers even in this pathological case.
The rest of the claim follows directly from the relevant definitions. 
\end{proof}

\subsection{The remaining
cases for defining the functions $N_{\sigma\tau}$}\label{definingfunctionsappendix}

\noindent
\textbf{Case 2.} We now assume (cf. \textbf{Case 1})
that \textbf{1.a} and \textbf{1.b} hold but \textbf{1.c} does not.
Now, if $\sigma$ and $\tau$ are \emph{inverses} of each other, then we define $N_{\sigma\tau}$ as
follows using $T_{\empty_{k+1,\ell}}$
from Equation \ref{twoconfunction22} of Proposition \ref{bipartitecounting22}:
\begin{multline}\label{needtoinventtoomanynames}
\scalebox{0.95}[1]{$N_{\sigma\tau}(\overline{n})\ :=
T_{\empty_{k+1,\ell}}(n_{\sigma\tau},n_{\sigma},n_{\tau\sigma},
n_{\tau},w_1,\dots , w_k,y,x_1,\dots ,x_{\ell}).$}
\end{multline}
If $\sigma$ and $\tau$ are \emph{not} inverses of each other, we
consider three subcases.
Firstly, if both $\sigma$ and $\tau$ are both ways witnessing, then we define
%
%
%
%First, if neither $\sigma$ nor $\tau$ is the
%inverse of the other $2$-type, we define $N_{\sigma\tau}$ as follows:
%
%
%
\begin{equation}\label{thehundrethequation}
N_{\sigma\tau}(\overline{n})
:=
\begin{cases}
0\ \ \text{ if }n_{\sigma\tau}\not=0\text{ or }n_{\tau\sigma}\not=0,\\
 L_{\empty_{k,\ell}}\bigl(\, 
n_{\sigma}\cdot n_{\tau},\,
w_1,\dots , w_k,x_1,\dots , x_{\ell}\bigr)\text{ otherwise}.
\end{cases}
\end{equation}
Secondly, if $\sigma$ is both ways witnessing but $\tau$ not, we define $N_{\sigma\tau}$ as follows,
letting $\overline{w}'$ denote the list $w_1,\dots , w_k,x_1,\dots , x_{\ell},y,z$ of weights:
\begin{multline}\label{needtoinventtoomanynames222222}
N_{\sigma\tau}(\overline{n})\ := 
\begin{cases}
0\ \ \text{ if }n_{\sigma\tau}\not=0,\\
P_{\empty_{k,\ell+2}}(n_{\sigma\tau},n_{\sigma},n_{\tau\sigma},
n_{\tau},\overline{w}'\, )\ \text{otherwise}.
\end{cases}
\end{multline}
%
%
%
%
%then define $N_{\sigma\tau}$ to be zero when $n_{\sigma\tau}\not= 0$ and otherwise as
%given by Equation \ref{needtoinventtoomanynames} but with $y$ 
%replaced by $z$.
%
%
%
Finally, the case where $\tau$ is both ways witnessing but $\sigma$ is not is analogous.
%$N_{\sigma\tau}$ is zero when $n_{\tau\sigma}\not=0$
%and otherwise $N_{\sigma\tau}$ is defined as in Equation \ref{needtoinventtoomanynames}.
%

%
\medskip
\noindent
\textbf{Case 3.} We assume that \textbf{1.a} holds but \textbf{1.b} not.
If both $\sigma$ and $\tau$ are incombatible with each other, we
define $N_{\sigma\tau}$ exactly as in Equation \ref{thehundrethequation}.
If $\tau$ is combatible with $\sigma$ but $\sigma$ not with $\tau$, we 
define $N_{\sigma\tau}$ as
follows, with two subcases.
Firstly, if $\tau$ is both ways witnessing, we define $N_{\sigma\tau}$ as in
Equation \ref{thehundrethequation}. If $\tau$ is \emph{not} both ways witnessing, we
define $N_{\sigma\tau}$ according to Equation \ref{needtoinventtoomanynames222222}.
%
%
%
%follows, letting $\overline{w}'$ denote the list $w_1,\dots , w_k,x_1,\dots , x_{\ell},y,z$ of weights:
%
%
%
\begin{comment}
\begin{multline}\label{needtoinventtoomanynames23}
N_{\sigma\tau}(\overline{n},\overline{w})\ := 
\begin{cases}
0\ \ \text{ if }n_{\tau}\not=0,\\
P_{\empty_{k,\ell+2}}(n_{\sigma\tau},n_{\sigma},n_{\tau\sigma},
n_{\tau},\overline{w}'\, )\ \text{otherwise}.
\end{cases}
%
%
%
\end{multline}
%
%
%
\begin{multline}\label{needtoinventtoomanynames234}
N_{\sigma\tau}(\overline{n},\overline{w})\ := 
\begin{cases}
0\ \ \text{ if }n_{\tau\sigma}\not=0,\\
P_{\empty_{k,\ell+2}}(n_{\sigma\tau},n_{\sigma},n_{\tau\sigma},
n_{\tau},\overline{w}'\, )\ \text{otherwise}.
\end{cases}
%
\end{multline}
\end{comment}
%
%
%
The case where $\sigma$ is combatible with $\tau$
but $\tau$ not with $\sigma$, is analogous.
\medskip
\medskip
\noindent
\textbf{Case 4.}
We assume that \textbf{4.a)} $\sigma=\tau$;
\textbf{4.b)} $\sigma$ is compatible with itself, meaning that the
first and second $1$-types of $\sigma$ are the same; \textbf{4.c)} $\sigma$ is 
\textcolor{black}{not} both ways witnessing.
By Equation \ref{twoconfunctionone} in Proposition \ref{internalcounting1},
the weight contributed by edges from $B_{\sigma}$ to $B_{\sigma}$ itself is
thus given by
\begin{equation}\label{199919191}
N_{\sigma\sigma}(\overline{n})\ :=\
J_{\empty_{k,\ell+1}}(n_{\sigma\sigma},n_{\sigma},w_1,\dots , w_k,x_1,\dots ,x_{\ell},y),
\end{equation}
which defines the function $N_{\sigma\sigma}$ in this particular case. When \textbf{4.a}
and \textbf{4.b} hold but \textbf{4.c} not, so $\sigma$ is
both ways witnessing, then we define,
using the function $S_{\empty_{k+1,\ell}}$ of Equation \ref{dunnohowmaniethequation} in
Proposition \ref{internalcounting2}, that
\begin{equation}\label{randomverticalspace}
N_{\sigma\sigma}(\overline{n})\ :=\
S_{\empty_{k+1,\ell}}(n_{\sigma\sigma},n_{\sigma},w_1,\dots , w_k,y,_1,\dots ,x_{\ell}).
\end{equation}
When \textbf{4.a} holds
but \textbf{4.b} not, we define $N_{\sigma\sigma}(\overline{n})$ to be
zero when $n_{\sigma\sigma}\not=0$
and otherwise as given by Equation \ref{randomverticalspace}.
%

%
\begin{comment}
All remaining cases are straighforward with the help of the
results provided in the previous section.
\end{comment}
By observing
that the expressions in Equation \ref{ultimatecountingalg} can 
easily be computed in $\mathrm{PTIME}$, we obtain the theorem
that the weighted model counting problem for each sentence of two-variable logic with a
functionality axiom is in $\mathrm{PTIME}$.

\subsection{Proof of Lemma \ref{utosu}}\label{proofofSUvsUlemma}

Before Proving Lemma \ref{utosu}, we make the following auxiliary definitions.

An \emph{identity literal} is an atom $x=y$ or 
negated atom $x \not= y$.
An \emph{identity profile} $\varphi$ over a set $X$ of variables is a consistent
conjunction with precisely one of the literals $x=y$, $x\not=y$ 
for each two distinct variables $x,y\in X$; consistency of $\varphi$
means that $\varphi\not\models\bot$.
Note that the formula $\mathit{diff}(x_1,\dots , x_k)$ is the identity profile over $\{x_1,\dots , x_k\}$
where all identities are negative. An identity profile $\varphi$ is \emph{consistent
with} a conjunction $\psi$ of identity literals if $\varphi\wedge\psi\not\models \bot$.

We then prove Lemma \ref{utosu}:

\begin{proof}
We will prove the equivalent claim that every $\exists^*$-sentence of $\mathrm{U}_1$ 
translates to an equivalent Boolean combination of $\exists^*$-sentences of $\mathrm{SU}_1$.
Thus we fix a $\mathrm{U}_1$-sentence $\exists x_1\dots \exists x_{\ell}\, \psi$
where $\psi$ is quantifier-free.
%
%We define $X := \{x_1,\dots ,x_k\}$ and
We let $\eta$ be the vocabulary of $\psi$.
As $\psi$ is a $\mathrm{U}_1$-matrix, all the
higher arity atoms of $\psi$ have
the same set $Y\subseteq \{x_1,\dots , x_{\ell}\}$ of variables.
We let $Y:= \{y_1,\dots ,y_k\}$.
%By Lemma \ref{qfrtypeslemma}, $\psi$ is
%equivalent to a disjunction $\psi_1\vee\dots \vee\psi_l$ of
%quantifier-free $(Y,X)$-types for $\mathrm{U}_1$ over $\tau$.
%
%
%

%
%
%
We then begin modifying the
sentence $\exists x_1\dots \exists x_{\ell}\psi$.
We first put $\psi$ into 
disjunctive normal form, thereby obtaining
an equivalent sentence 
$\exists x_1\dots \exists x_{\ell}(\psi_1\vee\dots \vee \psi_m),$
where each formula $\psi_i$ is
free of disjunctions.
%(possibly negated) equalities, unary literals and $Y$-literals.
%Let $X:=\{y_1,\dots ,y_m\}$.
%
%
%
We then distribute the quantifier block $\exists x_1\dots \exists x_{\ell}$ of
over the disjunctions, 
obtaining the sentence
$\chi := (\exists x_1\dots \exists x_{\ell}\psi_1)\ \vee\dots
\vee (\exists x_1\dots \exists x_{\ell}\psi_{m}).$ 
Now, let us fix a disjunct $\exists x_1\dots \exists x_{\ell}\psi_i$ of $\chi$.
To conclude the proof, it
suffices to show that $\exists x_1\dots \exists x_{\ell}\psi_i$ is
equivalent to a Boolean combination of $\exists^*$-sentences of $\mathrm{SU}_1$.
We assume, w.l.o.g., that
$\psi_i\, :=\, \chi_{\mathit{id}}\ \wedge\ \chi_1\ \wedge\ \chi(y_1,\dots y_k),$
where $\chi_{\mathit{id}}$ is a conjunction of identities and negated 
identities, $\chi_1$ a conjunction of unary literals, and $\chi(y_1,\dots ,y_k)$ a
conjunction of $Y$-literals. We also assume, w.l.o.g., that 
the vocabulary of $\psi_i$ is $\eta$ and that $\psi_i$ is
consistent, i.e., $\psi_i\not\models\bot$. If $\psi_i$ was not consistent, we
would be done with the proof, as $\exists x_1\dots \exists x_{\ell}\psi_i$ 
would be equivalent to $\bot$. 
Let $I$ denote set of identity profiles over $\{x_1,\dots , x_{\ell}\}$
consistent with $\chi_{\mathit{id}}$. Thus
$$\psi_i\, \equiv\, \bigvee\limits_{\gamma\, \in\, I}(\gamma\, \wedge\, \chi_1\,
\wedge\, \chi(y_1,\dots y_k)).$$
%
%
%
\begin{comment}
Similarly, let $T$ be the set of $k$-tables $\beta(y_1,\dots , y_k)$ over $\eta$  
such that $\beta(y_1,\dots , y_k)\models \chi(y_1,\dots ,y_k)$. We have 
%
%
%
$$\psi_i\, \equiv\, \bigvee\limits_{\substack{\gamma\, \in\, \vspace{1mm} I\\ \beta\, \in\, T}}(
\gamma\, \wedge\, \chi_1\,
\wedge\, \beta).$$
%
%
%
%
%
%
%
\end{comment}
%
%
%
Now recall that $\chi_1$ is a conjunction of unary literals whose
variables are contained in $\{x_1,\dots , x_{\ell}\}$.
Thus $\chi_1$ is equivalent to a disjunction of 
conjunctions $\alpha_1(x_1)\wedge\dots \wedge \alpha_{\ell}(x_{\ell})$
where each $\alpha_i$ is a $1$-type over $\eta$. 
Let $A$ denote the set of all such conjunctions. Thus we have 
$$\psi_i\, \equiv\, \bigvee\limits_{(\varphi,\gamma)\, \in
\, A\, \times\, I}(\gamma\, \wedge\, \varphi\
\wedge\, \chi(y_1,\dots , y_k)).$$
Now, in order to obtain a suitably modified variant of the 
sentence $$\exists x_1\dots \exists x_{\ell}\, \psi_i,$$ we
distribute the block $\exists x_1\dots \exists x_{\ell}$ of quantifiers
over the disjunctions of the right hand side of the above equation and
thereby observe that
%
%
%
%$$\exists x_1\dots \exists x_{\ell}\, \psi_i\, \equiv\, \exists x_1\dots \exists x_{\ell}
%\bigvee\limits_{(\varphi,\beta,\gamma)\, \in\, A\, \times\, T\, \times\, I}(\gamma\, \wedge\, \varphi\,
%\wedge\, \beta).$$
%
%
%
%We distribute the block $\exists x_1\dots \exists x_{\ell}$ of quantifiers
%over the disjunctions, whence
%
%
%
\begin{multline*}
\exists x_1\dots \exists x_{\ell}\, \psi_i\\
\equiv\, \bigvee\limits_{(\varphi,\gamma)\ \in\ A\, \times\, I}
\exists x_1\dots \exists x_{\ell}
(\gamma\, \wedge\, \varphi\, \wedge\, \chi(y_1,\dots , y_k)).
\end{multline*}
We fix a single
disjunct $\delta := \exists x_1\dots \exists x_{\ell}(\gamma\, \wedge\, \varphi\,
\wedge\, \chi(y_1,\dots , y_k))$ and show how to
translate it to a Boolean combination of $\exists^*$-sentences of $\mathrm{SU}_1$,
thereby concluding the proof.
%

%
%of $\exists x_1\dots \exists x_k(\gamma\, \wedge\, \alpha\,
%\wedge\, \beta)$
If $\gamma$ contains non-negated identities, 
we eliminate them by renaming variables in the 
quantifier-free part of $\delta$.
%and removing the positive identities from $\gamma$.
Thus we obtain a sentence $\delta' := 
\exists z_1\dots \exists z_{n}(\gamma'\, \wedge\, \varphi'\,
\wedge\, \chi')$ equivalent to $\delta$ such that the following conditions hold.
%
%
%
%\begin{enumerate}
%

\medskip

\noindent
\textbf{1.)}
$\{z_1,\dots , z_n\}\subseteq \{x_1,\dots , x_{\ell}\}$
and $\gamma'$ is the formula $\mathit{diff}(z_1,\dots, z_{n})$ 
(which is simply $\top$ if $n = 1$).

\medskip

\noindent
\textbf{2.)} $\varphi'$ is a conjunction $\alpha_1(z_1)\wedge \dots \wedge \alpha_n(z_n)$
of $1$-types  containing 
\emph{at least one} $1$-type for each variable $z_1,\dots , z_{n}$; if the
conjunction has two or more types for the same
variable, then it is inconsistent, and thus $\delta'\equiv \bot$, so
we are done with the proof. Therefore we assume that $\alpha_1(z_1)\wedge \dots \wedge \alpha_n(z_n)$
has exactly one $1$-type for each variable.

\medskip

\noindent
\textbf{3.)} $\chi'$ is a conjunction of $Z$-literals for
some set $Z\subseteq \{z_1,\dots , z_{n}\}$ of
variables. We assume, w.l.o.g., that $Z = \{z_1,\dots , z_{m}\}$
for some $m\leq n$. We note the following.
\begin{enumerate}
\item[\textbf{3.a)}]
It is possible that the variable renaming process makes $\chi'$
inconsistent, as for example when $Rxyz,\neg Ryxz$ are
replaced by $Ryyz,\neg Ryyz$.
If $\chi'$ is inconsistent, we are
done with the proof. Thus we assume that $\chi'$ is consistent.
\item[\textbf{3.b)}]
There exists a disjunction $\beta_1\vee\dots \vee \beta_p$
of $|Z|$-tables such that $\beta_1\vee\dots \vee \beta_p\models\chi'$.
\end{enumerate}
%
%
%
%\end{enumerate}
%
%
%
Thus $\delta'\equiv \bigvee_{i\leq p}\exists z_1
\dots \exists z_{n}(\gamma'\, \wedge\, \varphi'\,
\wedge\, \beta_i)$, where each $\beta_i$ is a $|Z|$-table.
Once again distributing 
the quantifiers, we get a disjunction $\exists z_1
\dots \exists z_{n}(\gamma'\, \wedge\, \varphi'\,
\wedge\, \beta_1) \vee \dots \vee \exists z_1
\dots \exists z_{n}(\gamma'\, \wedge\, \varphi'\,
\wedge\, \beta_p)$. It suffices to fix one of 
these disjuncts $\exists z_1\dots \exists z_{n}(\gamma'\, \wedge\, \varphi'\,
\wedge\, \beta_i)$ and show how it translates into a 
Boolean combination of $\exists^*$-sentences of $\mathrm{SU}_1$.
Now, $\exists z_1\dots \exists z_{n}(\gamma'\, \wedge\, \varphi'\,
\wedge\, \beta_i)$ is the sentence

\medskip

\scalebox{0.9}[1]{$\exists z_1\dots \exists z_{n}\bigl(\mathit{diff}(z_1,\dots , z_{n})\wedge
\alpha_1(z_1)\wedge \dots \wedge \alpha_{n}(z_{n})
\wedge \beta_i(z_1,\dots , z_m)\bigr).$}

\medskip

\noindent
Since each element of a model must
satisfy exactly one $1$-type, we observe
that this sentence is equivalent to the following sentence (where
the first main conjunct has $m$ and the second one $n$ variables):
\begin{multline*}
\scalebox{0.86}[1]{$\exists z_1\dots \exists z_m \bigl(\,
\alpha_1(z_1)\wedge\dots \wedge \alpha_m(z_m)
\wedge \beta_i(z_1,\dots , z_m) \wedge \mathit{diff}(z_1,\dots , z_m)\, \bigr)$}\\
\scalebox{0.86}[1]{$\wedge\ \exists x_1\dots \exists x_{n}\bigl(\, 
\alpha_1(x_1)\wedge\dots \wedge \alpha_n(x_{n})
\wedge \mathit{diff}(z_1,\dots , z_{n})\bigr).$}
\end{multline*}
Both of these conjuncts are $\mathrm{SU}_1$-sentences.
\end{proof}
\section{Appendix: Proof of Proposition \ref{prefixclassproposition}}\label{prefixclasses}
\begin{proof}
The second claim of Proposition \ref{prefixclassproposition} is 
immediate, as \cite{DBLP:conf/kr/BroeckMD14} shows that the 
symmetric weighted model counting problem is in $\mathrm{PTIME}$ for
each formula of  two-variable logic. We thus turn to the first claim. 
The article \cite{suc15} provides a sentence $\varphi$ of 
three-variable logic $\mathrm{FO}^3$ that has a $\#\mathrm{P}_1$-complete
symmetric weighted model counting problem.
Given $\varphi$, there is a very simple way to \emph{directly} 
ensure that there exists a sentence $\forall x\forall y \forall z\psi$
(where $\psi$ quantifier-free) with the same
model counting problem as $\varphi$. The prodecure is straightforward and
interesting in its own right.
The idea is to process $\varphi$ in a way that bears a resemblance to 
the Scott normal form reduction. We describe the procedure for an
arbitrary $\mathrm{FO}^3$-sentence $\chi$.
We begin eliminating quantifiers of quantified subformulae of $\chi$, one
quantifier at a time, starting from the atomic level and working our way upwards.
Consider a subformula $\chi'(y,z) := Qx\chi_0(x,y,z)$ of $\chi$
where $Q\in\{\forall, \exists\}$
and $\chi_0$ is quantifier-free.
Now, $\chi'(y,z)$ has two free variables.
Therefore we let $P_{\chi'}$ be a fresh binary 
predicate and consider the sentence 
$$\forall y \forall z
( P_{\chi'}(y,z) \ \leftrightarrow\ Q x\chi_0(x,y,z))$$
stating that $\chi'(y,z)$ is equivalent to $P_{\chi'}(y,z)$.
Letting $Q'$ denote the dual of $Q$, i.e., $Q' \in \{\exists,\forall\}\setminus \{Q\}$,
this sentence is easily seen equivalent to
\begin{multline}
\chi^*\ :=\ \forall y\forall z Qx ( P_{\chi'}(y,z)\ \rightarrow
\chi_0(x,y,z))\\
\wedge \forall y\forall z Q'x( \chi_0(x,y,z)\rightarrow P_{\chi'}(y,z)).
\end{multline}
Therefore $\chi$ is equivalent to the sentence 
%
%
%
%\begin{multline}
%
$$\chi''\ :=\ \chi^*\ \wedge\ \chi[P_{\chi'}(y,z)/\chi'(y,z)],$$
%
%\end{multline}
%
%
%
where $\chi[P_{\chi'}(y,z)/\chi'(y,z)]$ is
obtained from $\chi$ by replacing the formula $\chi'(y,z)$ with $P_{\chi'}(y,z)$.
Here $Qx\chi_0(x,y,z)$ had two free variables, but we may also need to 
eliminate quantifiers $Qx$ from formulae of type $Qx\varphi'$ with one or zero
free variables; here $\varphi'$ is quantifier-free.
The elimination is then, however, done in a similar way, the main 
difference being that the fresh predicate then has arity one or zero. 
Repeating the procedure, we eliminate quantifiers one by one, starting from the atomic
level and working upwards from there. We ultimately end up
with a conjunction
$$\overline{Q}_1\chi_1
\wedge\dots \wedge \overline{Q}_k\chi_k$$
where each $\overline{Q}$ is a block of three 
quantifiers (introducing dummy quantifiers if necessary), while
each $\chi_i$ is quantifier-free. Now, similarly to the case
with Scott normal form reductions discussed above, the novel 
predicates have been axiomatized to have a unique interpretation in any 
model that satisfies $\varphi: = \overline{Q}_1\chi_1
\wedge\dots \wedge \overline{Q}_k\chi_k$, and thus an
analogous result to Lemma \ref{scottlemma2} holds:
\begin{equation}\label{dunnowhocares}
\mathrm{WFOMC}(\chi,n,w,\bar{w})
= \mathrm{WFOMC}(\varphi,n,w',\bar{w}'),
\end{equation}
where $w'$ and $\bar{w}'$
extend $w$ and $\bar{w}$ by sending the
novel symbols to $1$.
Then we apply the Skolemization procedure (Lemma \ref{skolemlemma}) to $\varphi$,
%$\mathrm{WFOMC}(\chi,n,w,\bar{w})
%= \mathrm{WFOMC}(\varphi,n,w',\bar{w}')$,
thus obtaining a conjunction of the form $$\forall x \forall y \forall z\chi_1'
\wedge\dots \wedge \forall x \forall y \forall z \chi_k'.$$
We combine the matrices $\chi_1',\dots , \chi_k'$ under the 
same quantifier prefix, thus obtaining a
sentence $\varphi' := \forall x\forall y\forall z\, \gamma$, where $\gamma$ is
quantifier-free. We now have $$\mathrm{WFOMC}(\varphi,n,w',\bar{w}')
= \mathrm{WFOMC}(\varphi',n,w'',\bar{w}''),$$ where $w''$ and $\bar{w}''$
treat the fresh symbols as specified in Lemma \ref{skolemlemma}, $w''$
mapping them to $1$ and $\bar{w}''$ to $-1$. Thus, combining this 
with Equation \ref{dunnowhocares}, we obtain
$$\mathrm{WFOMC}(\chi,n,w,\bar{w})
= \mathrm{WFOMC}(\varphi',n,w'',\bar{w}'').$$
Thus we finally obtain the sentence $\varphi'$ with
prefix $\forall\forall\forall$ with a $\#\mathrm{P}_1$-complete
weighted model counting problem.

We then start modifying the $\forall\forall\forall$-sentence $\varphi'$ in order to obtain, for
each prefix class $C$ with three quantifiers, a
sentence in $C$ with the same weighted model counting problem as $\varphi'$.
The required modifications will be made by operations that slightly modify the
Skolemization operation from Section \ref{sect:skolemn}. We define these operations next. 
Let $\chi' := \forall x_1\dots \forall x_k Q_1x_{k+1}\dots Q_m x_m\, \chi''$ be a 
prenex normal form sentence with $\chi''$ quantifier-free and with $Q_i\in\{\exists,\forall\}$ for
each $i$. We turn $\chi'$ into $\forall x_1\dots \forall x_k Q_1'x_{k+1}\dots Q_m'
x_m (Ax_1\dots x_k\vee \neg\chi'')$,
where $A$ is a fresh $k$-ary predicate and each $Q_i'$ is the dual of $Q_i$.
The difference with the Skolemization operation of
Section \ref{sect:skolemn} is simply that $Q_1$ is not required to be $\exists$.
This new sentence has the same
model counting problem as $\chi''$ when the fresh symbol $A$ is 
given weights exactly as in Lemma \ref{skolemlemma}. The proof of this claim
almost the same as the
proof of Lemma \ref{skolemlemma}, which is given in Appendix \ref{skolemlemmaproof}.
Notice that we can have $k=0$, and then the new predicate $A$ is nullary.
With these novel Skolemization operations, we can take any 
prenex normal form sentence and modify it so that the original 
prefix $\overline{\forall}\, Q_1\dots Q_m$ changes to $\overline{\forall}\, Q_1'\dots Q_m'$
where $\overline{\forall}$ is in both cases the same (possibly empty) string of universal 
quantifiers and $Q_1'\dots Q_m'$ is obtained from $Q_1\dots Q_m$ by changing each 
quantifier to its dual. It remains to show that with these simple operations, we 
can obtain from the prefix $\forall\forall\forall$ all the
remaining seven prefixes with three quantifiers.
We obtain $\exists\exists\exists$ from $\forall\forall\forall$ by letting all 
the three quantifiers in $\forall\forall\forall$ be the suffix that
gets dualized. We get $\forall
\exists\exists$ from $\forall
\forall\forall$ by dualizing the last two universal quantifiers. Similarly, we get $\forall
\forall\exists$ from $\forall\forall\forall$ by dualizing the last universal quantifier.
Now, having $\forall\forall\exists$, we
obtain $\forall\exists\forall$ by dualizing the last two quantifiers
and $\exists\exists\forall$ by dualizing all the three quantifiers.
From $\forall\exists\forall$, we then get $\exists\forall\exists$ by 
dualizing all quantifiers. Finally, from $\forall\exists\exists$ obtained
earlier on, we get the last remaining prefix $\exists\forall\forall$ by dualizing everything.
Thus we have shown that all prefix classes with at least three quantifiers
have a sentence with a $\#\mathrm{P}_1$-complete symmetric
weighted model counting problem. Together with the
first claim of Proposition \ref{prefixclassproposition},
this gives the desired complete classification of
first-order prefix classes.
\end{proof}

\end{document}